\documentclass[acmsmall,nonacm]{acmart}

\usepackage{algorithm}
\usepackage{algorithmic}

\newcommand{\li}{[\![}
\newcommand{\ri}{]\!]}

\newtheorem{prop}{Proposition}
\newtheorem{thm}{Theorem}
\newtheorem{prob}{Problem}

\newtheorem{eg}{Example}
\newtheorem{defn}{Definition}
\newtheorem{rmk}{Remark}
\newtheorem{lem}{Lemma}

\newcommand{\liren}[1]{{\color{black} #1}}

\AtBeginDocument{%
  \providecommand\BibTeX{{%
    \normalfont B\kern-0.5em{\scshape i\kern-0.25em b}\kern-0.8em\TeX}}}

\acmJournal{TECS}
\begin{document}

\title{Synthesis-guided Adversarial Scenario Generation for Gray-box Feedback Control Systems with Sensing Imperfections}
 
\author{Liren Yang}
\email{yliren@umich.edu}
% \orcid{0000-0002-7677-8543}
\author{Necmiye Ozay}
\email{necmiye@umich.edu}
\affiliation{%
  \institution{University of Michigan}
  % \streetaddress{P.O. Box 1212}
  \city{Ann Arbor}
  \state{Michigan}
  \country{USA}
  \postcode{48105}
}

\thanks{This article appears as part of the ESWEEK-TECS special issue and was presented in the International Conference on Embedded Software (EMSOFT), 2021. 
DOI: \url{https://doi.org/10.1145/3477033}
}

\renewcommand{\shortauthors}{Yang and Ozay}
\renewcommand{\shorttitle}{Synthesis-guided Adversarial Scenario Generation for Gray-box Systems with Sensing Imperfections}

\begin{abstract}
 In this paper, we study  feedback dynamical systems with memoryless controllers under imperfect information. 
We develop an algorithm that searches for ``adversarial scenarios", which can be thought of as the strategy for the adversary representing the noise and disturbances, that lead to safety violations. The main challenge is to analyze the closed-loop system's vulnerabilities with a potentially complex or even unknown controller in the loop. As opposed to commonly adopted approaches that treat the system under test as a black-box, we propose a synthesis-guided approach, which leverages the knowledge of a plant model at hand. This hence leads to a way to deal with gray-box systems (i.e., with known plant and unknown controller). Our approach reveals the role of the imperfect information in the violation. Examples show that our approach can find non-trivial scenarios that are difficult to expose by random simulations.   This approach is further extended  to incorporate model mismatch and to falsify vision-in-the-loop systems  against finite-time reach-avoid specifications.
\end{abstract}




\renewcommand\footnotetextcopyrightpermission[1]{}
\pagestyle{plain}

\maketitle

\section{Introduction}

\begin{figure}[t]
    \includegraphics[width=0.45\textwidth]{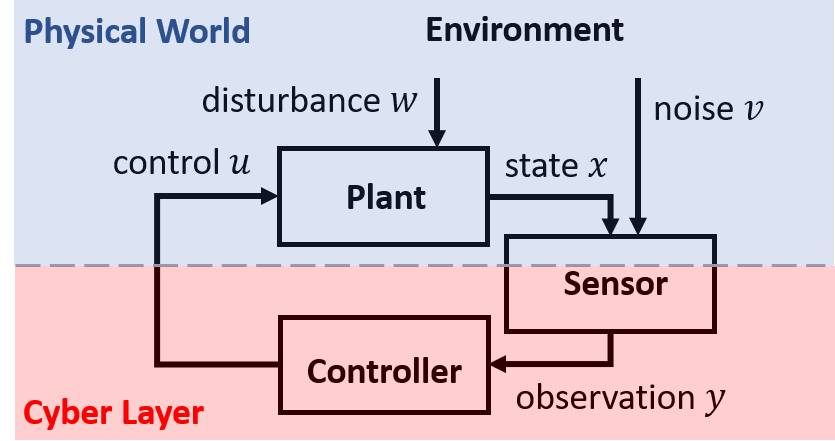}
	\vspace{-0mm}
      \caption{Block diagram of a closed-loop control system.}
    \label{fig:blk_0}
\end{figure} 
Cyber-physical systems (CPS) combine physical processes and computational (cyber) elements to perform complicated tasks in a dynamic environment.  
From a control perspective, a simplified CPS can be described by the block diagram in Fig. \ref{fig:blk_0}, where the plant represents the physical process, and the controller makes decisions at the cyber layer. 
The disturbance input $w$ encompasses environment parameters, external events and other agents' actions that directly act on the plant, whereas the input $v$ captures the noise and data loss, which lead to imperfect information $y$ and affect the controller's decision. 
The controller's goal is to determine control actions $u$ so that the plant's state trajectories satisfy some high-level specifications under all possible environmental uncertainties.

One key challenge in the CPS research is, given a control design, how to prove that it achieves desired closed-loop behavior. 
This question is particularly important for safety-critical systems. 
Due to the complexity of the specification and the environment, the controller may be complicated even if the plant has modest complexity and may include optimization, learning and vision based components.  The formal verification of the closed-loop system is hence a difficult task. 
Instead of verification, one alternative is to start with the plant model and close the loop by directly synthesizing a correct-by-construction controller, which achieves the specification provably. 
However, the synthesis problem is also challenging, especially when the controller does not have perfect information of the state of the world \cite{de2006lattice}, \cite{majumdar2020abstraction}, \cite{yang2020efficient}. For the cases where a vision-based solution must be used, such task is even harder. 
Other than synthesis, a more practical alternative is to quickly find scenarios under which the system trajectory violates the specification, i.e., to falsify the system. 
These adversarial scenarios are valuable for understanding the shortcomings of the controller at  early design stages, which may be hard to expose by random simulations. 
Moreover, once found, these adversarial scenarios can be used to improve the design, e.g., see \cite{fremont2020formal}, \cite{ghosh2019counterexample},  \cite{wang2020falsification}.

The problem of finding adversarial scenarios (a.k.a. falsification) for CPS has attracted much attention in the past two decades.  
Many approaches and tools are developed using stochastic search and optimization techniques, e.g., random tree search \cite{bhatia2004incremental}, \cite{ernst2019fast}, \cite{kim2005rrt}, \cite{nahhal2007test}, \cite{plaku2009falsification}, specification-guided stochastic sampling \cite{annpureddy2011s}, Tabu search \cite{deshmukh2015stochastic}, Bayesian optimization  \cite{deshmukh2017testing}, \cite{ghosh2018verifying}, nonlinear simplex optimization \cite{donze2010breach}. 
These works all treat the system under test as a black-box to avoid its high complexity, and search for a specification violation at the system level.  
Notably, some recent works \cite{dreossi2019compositional}, \cite{dreossi2019verifai}, \cite{tuncali2018simulation} study systems with vision/learning-based controllers in the loop. 
In particular, \cite{tuncali2018simulation} searches for a physical environment configuration (encoded by $w$) at the system level, whereas \cite{dreossi2019compositional} uses a decomposition idea and focuses on the impact of imperfect information (encoded by $v$) to the learning component in the cyber layer, which is equivalently important.

One closely related domain is adversarial learning, which includes a line of research that analyzes the vulnerability and robustness of machined learning algorithms against malicious agents. 
It is shown that inputs carefully crafted  by adding small perturbations unnoticeable to human eyes can fool well-trained neural network classifiers \cite{dalvi2004adversarial}, \cite{goodfellow2014explaining}, \cite{moosavi2016deepfool}, \cite{szegedy2013intriguing} or policies obtained by reinforcement learning \cite{huang2017adversarial}, \cite{kos2017delving}, and may harm the performance of practical CPS  (e.g., autonomous cars)  with deep neural networks in the loop \cite{pei2017deepxplore}. 
In the context of CPS falsification, these works focus on shifting the output of the learning-based controller by twisting the noise $v$ adversarially.  
As pointed out in \cite{tuncali2018simulation}, such attacks are at component level and do not necessarily lead to specification violations at the system level. This is because i) the effectiveness of the attack relies on the learning module being well-trained, and ii) a deviation in the learning algorithm's output does not necessarily lead to bad behaviors (e.g., to avoid an obstacle in the front, it is sometimes equivalently good to turn left or right). 
Adversarial learning methods usually require the component under analysis to be a white-box, and the black-box attacks are usually based on the principle of transferability \cite{bose2020adversarial}, \cite{huang2017adversarial}.  
More recently, a slightly different work \cite{gleave2019adversarial} presents a way to falsify reinforcement learning policies at the system level using  the disturbance $w$, which captures the move of an adversarial agent. However, the form of the entire closed-loop system must be known.

In this paper, we ask the following question: when there is a control-oriented plant model with modest complexity at hand, yet the closed-loop system is still hard to formally analyze due to highly complex controllers, is there a way to leverage the knowledge of this simple plant model and quickly identify non-trivial safety violations at the system level, where both the disturbance $w$ and the noise $v$ are essential for the violations to occur?  
Here, by a ``simple plant'', we mean any open-loop system (e.g., switched-affine systems) for which a two-player safety game can be solved relatively efficiently. 
By a ``complex controller'', we mean any controller  (e.g., hybrid MPC, neural network controllers or black-box controllers) that makes it challenging to compute the backward reachable set for the closed-loop system.  
This problem is motivated by systems such as autonomous cars, which have relatively simple motion dynamics but complex controllers and imperfect sensors (e.g., a camera). 
We refer to such systems with known plant models and unknown controllers as gray-box systems. 
Despite the same terminology, our setting is different  from that in \cite{yaghoubi2019gray}, where the term ``gray-box system'' refers to an unknown system whose closed-loop model can be approximated with a local linear model constructed on the fly in the falsification process. 
Another related work is \cite{waga2020falsification}, where a finite abstraction (Mealy machine) is learned from a CPS, and adversarial scenarios of the CPS is searched via model checking on the abstraction.  
The main difference of our setting is the separation of the plant and the controller. 
Table \ref{tab:compare} compares our setting and the aforementioned works. 

We focus on safety specifications and static output-feedback controllers, and propose a synthesis-guided falsification approach to find adversarial scenarios\footnote{The term ``falsification" and ``adversarial scenario generation" are used interchangeably from now on.}. The key idea is to inner-approximate the closed-loop system's backward reachable set under the guidance of a local two-player game's solution, whose complexity is independent of the controller. 
We further approach the falsification of vision-in-the-loop systems from such imperfect information game standpoint and show that a semantic space search reduction similar to \cite{dreossi2019compositional}, \cite{tuncali2018simulation} can be used in our framework.

\begin{table}[t]
  \caption{Adversarial Learning vs. Falsification Comparison}
  \label{tab:compare}
\centering
{\small	\begin{tabular}{c || c | c | c}
	\hline
	  & White-box &  Falsifying &  Falsification \\
	  & Components  &  Input &  Level\\
	\hline
	\hline
	 Adversarial & controller & $v$ & component \\
      \cline{2-4} 
	\  Learning &\cite{gleave2019adversarial}: plant \& controller \ & $w$  &  system \\
	\hline
	 Falsification & $-$ & $w$ \& $v$ & system \\
	\hline
	This Paper & plant & $w$ \& $v$  & system \\
	\hline
	\end{tabular} 
}
\end{table} 

In the rest of the paper, we first review the problems of falsification and synthesis under full information in Section \ref{sec:prelim}, and discuss the falsification problem under imperfect information in Section \ref{sec:prob}. Then we present the following main contributions of our work. 
\begin{itemize}
\item In Section \ref{sec:sol}, we propose a synthesis-guided falsification framework to find adversarial scenarios under our gray-box settings. In this framework, different backward reachability analysis techniques can be applied (depending on the structural properties the plant has). We provide two implementations of our approach for plants with switched affine dynamics, one uses MPT3 toolbox \cite{MPT3} (works for systems up to 4D),  and the other is based on zonotope computation techniques developed in \cite{sadraddini2019linear} (works for systems of 10D or even higher). 
\item In Section \ref{sec:ext}, we extend our approach to incorporate model mismatch (Section \ref{sec:ext_mismatch}) and to falsify vision-in-the-loop systems (Section \ref{sec:ext_vision}) against finite-time reach and avoid specifications (Section \ref{sec:ext_reach}).  
\item In Section \ref{sec:eval}, we evaluate our approach and its extensions with several examples. 
We show that the approach can falsify black-box controllers, including hybrid model predictive controllers (MPC) and  feedforward neural network controllers, by running the same algorithm blindly. 
To demonstrate the scalability and effectiveness, we also test our approach against randomly generated problem instances. Experiments show that the obtained scenarios are nontrivial. 
\end{itemize}

\section{Preliminaries}\label{sec:prelim}

\textbf{Notations}. 
Let $S$ be a set, we use $S^\ast$ and $S^\omega$ to denote the set of finite and infinite sequences of elements in $S$, respectively. The set of finite sequences in $S^\ast$ of length $T$ is denoted by $S^T$, and the empty sequence of length zero is denoted by $\varepsilon$. 
Throughout the paper,  we will use bold font letters, e.g., $\textbf{s}$, to represent infinite sequences. 
We denote by $s_t$ the $(t+1)^{\rm th}$ element ($t$ starts from $0$) of a sequence $\textbf{s}$, and define $\textbf{s}_t = s_0 s_1 s_2 \dots s_t$ to be the prefix of the sequence $\textbf{s}$ until time $t$.
Let $\mathbb{Z}$ be the set of integers,  given $a\in \mathbb{Z}\cup\{-\infty\}$ and $b\in \mathbb{Z}\cup\{\infty\}$ such that  $a\leq b$, the set $\{s \in \mathbb{Z}\mid a\leq s \leq b\}$ is denoted by $\li a,b\ri$.

\subsection{Falsification \& Control Synthesis under Perfect Information}\label{sec:prelim_game}
Consider a discrete-time  time-invariant system $\overline{\Sigma}: x_{t+1} = f(x_t, u_t, w_t)$ 
where $x \in X$ is the state, $u\in U$ is the control input, $w \in W$ is the disturbance input (or process noise), 
and $f: X\times U \times W \rightarrow X$ is the transition map. 
A sequence $\textbf{x} \in X^\omega$ is a trajectory of the system under a \textit{state-feedback controller} $\overline{\pi}:X\rightarrow U$ if and only if (iff) there exists a control sequence $\textbf{u}\in U^\omega$ and a disturbance sequence $\textbf{w} \in W^\omega$ such that for all $t \in \li 0, \infty \ri$:  i) $x_{t+1} = f(x_t, u_t, w_t)$, and ii) $u_t = \overline{\pi}(x_t)$. 
With a slight abuse of terminology, we will also call the finite prefix $\textbf{x}_T$ a trajectory for any natural number $T$. 
Since the state-feedback controller $\overline{\pi}$ makes its decision based on the exact state $x_t$, it is said to have \textit{perfect information}. Also note that, due to the disturbance $w$, the closed-loop system's trajectory is not uniquely determined by the initial condition and the controller.

The following two safety problems  under perfect information have been studied in the literature: 
\begin{itemize}
\item \textbf{Falsification}: given a set $X_{\rm safe}$ of safe states, a set $X_{\rm init}$ of initial states and a state-feedback controller $\overline{\pi}: X\rightarrow U$, find a trajectory $\textbf{x}_T$ of the closed-loop system that starts at some initial state $x_0 \in X_{\rm init}$ 
% and a time instant $T$
 such that $x_T \notin X_{\rm safe}$ (proving that such $\textbf{x}_T$ does not exist is called safety \textbf{verification}). 
\item \textbf{Synthesis}: given a set $X_{\rm safe}$ of safe states, find a set of winning initial states and a state-feedback controller $\overline{\pi}: X\rightarrow U$ such that all the closed-loop trajectories $\textbf{x}$ under controller $\overline{\pi}$ starting from an arbitrary winning state will stay in the safe set for all time, i.e. $\textbf{x}\in X_{\rm safe}^\omega$. 
\end{itemize}

\subsection{Game Theoretical Interpretation of Verification/Falsification and Synthesis}
The above problems can be analyzed and solved with a game theoretical interpretation. 
The game is between the controller and the environment.
The controller's goal is to keep the states within the safe set $X_{\rm safe}$, while the environment aims to force the state out of $X_{\rm safe}$. 
At   time $t$, the controller first picks a control input $u_t$ based on the current state $x_t$, and then the environment picks a disturbance input $w_t$. At time $t+1$, the system will evolve to the new state $x_{t+1} = f(x_t, u_t, w_t)$, from where the game will proceed to the next round. 

% and a next state $x_{t+1} \in \tau(x_t, u_t)$ that is allowed by the uncertainty levels of the system model under the controller's decision, and this constitutes one round of the game. 
% The system then evolves to the new state $x_{t+1}$ and this finishes the round of the game. At time $t+1$, the game will proceed by starting the new round at state $x_{t+1}$. 

\textbf{Verification/falsification: one-player game}. 
The verification (or falsification) problem is a one-player game (because the control policy $\overline{\pi}$ is given) where we play the environment's role. 
This one-player game can be solved by performing the following iteration until $X_k^{\overline{\pi}}  \cap X_{\rm init} \neq \emptyset$, 
\begin{align}
    X_0^{\overline{\pi}} & = X_{\rm unsafe} : = X\setminus X_{\rm safe},\label{eq:iter11} \\
    X_{k+1}^{\overline{\pi}} & = \textbf{Pre}^{\overline{\pi}}(X_k^{\overline{\pi}}) := \big\{x \in X\, \big\vert \,\exists w \in W: f\big(x, {\overline{\pi}}(x), w\big) \in X_{k}^{\overline{\pi}} \big\}. \label{eq:iter12}
\end{align}
Once the sets $\{X_k^{\overline{\pi}}\}_{k=1}^T$ are computed, the falsifying run $\textbf{x}_T$ can be easily derived: 
\begin{align}
    x_0 & \in X_T^{\overline{\pi}} \cap X_{\rm init} \label{eq:fwd_e1}\\
    x_{t+1} & \in \big\{f\big(x_t, {\overline{\pi}}(x_t), w\big) \mid w \in W \big\} \cap X_{T-t - 1}^{\overline{\pi}} \text{ for } t\in \li 0, T-1\ri. \label{eq:fwd_e2}
\end{align}

\begin{rmk}
\normalfont
The above procedure can be split into two phases: the \textit{backward expansion phase}, during which we compute sets $\{X_k^{\overline{\pi}}\}_{k=0}^\infty$ by Eq. \eqref{eq:iter11}, \eqref{eq:iter12}, and the \textit{forward expansion phase} where the falsifying trajectory $\textbf{x}_T$ is computed via Eq. \eqref{eq:fwd_e1}, \eqref{eq:fwd_e2}. The computational effort is mainly spent on the backward expansion phase that involves manipulating sets, 
%  during which one needs to compute a set of states at each step, 
whereas the forward expansion phase is much cheaper because one only needs to find one state per step. 
In the backward expansion phase, the biggest challenge is to compute (or even to under approximate) the predecessors $\textbf{Pre}^{\overline{\pi}}$ of a set of states because the closed-loop dynamics is potentially complicated.  
Even when the open-loop dynamics $f$ is simple, this challenge may arise due of the complexity of the controller $\overline{\pi}$. 
\hfill\(\Diamond\)
\end{rmk}

\textbf{Falsification by synthesis: two-player game}. 
In the synthesis problem, we play the controller's role and refer to the game that we solve as the \textit{safety game}. 
The safety game is a two-player game because the environment's policy is not given a priori. 
If we play the role of the environment assuming that the controller $\overline{\pi}$ is unknown, the resulting two-player game is called the \textit{dual game} \cite{chou2018using} and can be used for falsification. 
We briefly describe the solutions to the safety game and the dual game and their connection below.

In the perfect information setting, there exists a maximal winning set of the safety game, i.e., the maximal robust controlled invariant set $R \subseteq X_{\rm safe}$.  
This maximal winning set can be achieved by a state-feedback controller $\overline{\pi}_{\rm best}$ such that $f\big(x, \overline{\pi}_{\rm best}(x), w\big)\subseteq R$ for all $x \in R$ and $w \in W$. 
The set $R$ can be computed at the limit by iteratively removing set $D_k$ from $X_{\rm safe}$, where $D_k$ is the set of states  that can be forced into the unsafe set $X_{\rm unsafe}: = X\setminus X_{\rm safe}$ in $k$ steps by the uncertainties. Formally, 
\begin{align}
    D_0 & = X_{\rm unsafe},  \label{eq:bkwd_e1}\\
    D_{k+1} & = \textbf{EPre}(D_k) : = D_k \cup \big\{x \in X \, \big\vert \, \forall u \in U: \exists w\in W: f(x, u, w)\cap D_{k} \neq \emptyset\big\}. \label{eq:bkwd_e2}
\end{align}
where $\textbf{EPre}(D_k)$ contains the \textit{environment-controllable predecessor} of set $D_k$.  
It is shown in \cite{bertsekas1972infinite} that under mild compactness conditions, 
\begin{align}
R= X\setminus D \text{ \ \ where \ \ }  D :=  \cup_{k=1}^{\infty} D_k. \label{eq:assump}
\end{align}
The set $D$ is called the \textit{dual game winning set under perfect information}, from where falsification is guaranteed regardless of the controller. Throughout this paper, we assume that Eq. \eqref{eq:assump} holds. 
Once we obtain sets $\{D_{k}\}_{k=0}^{\infty}$ by solving the dual game, an environmental strategy can be derived to pick the disturbance $w_t$, in a way to falsify a specific controller $\overline{\pi}$. Let $x_0 \in D =  \cup_{k=1}^{\infty} D_k$, define 
\begin{align}
x_0 & \in D_T, \text{ for some }T \label{eq:d_fwd_e1} \\
w_t & \text{ is s.t. } x_{t+1} = f\big(x_t, \overline{\pi}(x_t), w_t\big) \in D_{T-t - 1} \text{ for } t \in \li 0, T-1\ri \label{eq:d_fwd_e2}.
\end{align} 
It can be easily shown inductively that $x_t$ exists for all $t \in \li 1, T\ri$ and $x_{T} \in D_0 = X_{\rm unsafe}$. Hence $\textbf{x}_T$ is a trajectory falsifying the controller $\overline{\pi}$. 
Note that we need the above dual strategy for falsification because even if initiated in $D$, a safety violation is not guaranteed unless the environment picks the uncertainties adversarially against the given controller $\overline{\pi}$.

\begin{rmk}\label{rmk:generic}
\normalfont
The dual-game-based falsification procedure also consists of a backward expansion phase and a forward expansion phase. 
Notably, the complexity of the backward expansion only depends on
%  the complexity of
 the open-loop system's transition map $f$ and is independent of the controller $\overline{\pi}$. 
% it is not surprising that 
Via such falsification approach we only get \textit{generic adversarial scenarios}, in which case no controller can ensure safety 
starting from the falsifying initial condition $x_0$. 
\hfill\(\Diamond\)
\end{rmk}

\section{Adversarial Scenarios under Imperfect Information}\label{sec:prob}

In this section, we describe the problem of finding adversarial scenarios under imperfect information. 
We first introduce systems with imperfect information, then use an example to motivate finding ``controller-specific adversarial scenarios", and finally give the formal problem statement.

\subsection{Systems with Imperfect Information}
In this paper, we consider the safety control problem for systems with imperfect state information in the following form: 
\begin{align} 
\Sigma \ : \  
x_{t+1} & = f(x_t, u_t, w_t),  \ \ y_t  = g(x_t, v_t),  \label{eq:syst} 
\end{align}
where the state $x$, the control $u$, the disturbance $w$ and the transition map $f$ are defined exactly the same as the perfect information case,  $y\in Y$ is the output (observation), $v \in V$ is the measurement noise and $g: X \times V \rightarrow Y$ is the measurement map. 
We denote $x\sim y$ if there exists $v \in V$ such that $y = g(x,v)$. 
% The system is said to have perfect information if $Y = X$ and $\mu(x) = \{x\}$. 
The control input $u$ can be determined by an \liren{\textit{output-feedback control policy (controller)}} $\pi: (Y\times U)^\ast\times Y \rightarrow U$. We call a controller $\pi$ \textit{static} (aka, memoryless) if its decision at a time instant $t$ only depends on the latest observation $y_t$.

Given a controller $\pi$ and an initial state $x$, the closed-loop system's trajectory is not unique due to the system's  uncertainties. Hence falsifications, if possible, may not be achieved unless the environment adversarially picks the disturbance (i.e., process noise) $w_t$ and the sensor noise $v_t$. 
We introduce the following notion of falsifiability to capture the existence of an adversarial scenario. %  in the worst case. 

\begin{defn}\label{defn:fal}
\normalfont 
Let $\Sigma$ be a system described by Eq. \eqref{eq:syst} and  $\pi: (Y\times U)^\ast\times Y \rightarrow U$ be a controller. Given a state $x$, the closed-loop system starting from state $x$, denoted by $\Sigma_{\pi}(x)$, is called \textit{falsifiable} w.r.t. a safe set $X_{\rm safe}$ if there exists $T \in \mathbb{Z}$, $\textbf{x}_T\in X^{T+1}$, $\textbf{v}_{T-1} \in V^{T}$, $\textbf{y}_{T-1}\in Y^{T}$, $\textbf{u}_{T-1} \in U^{T}$, $\textbf{w}_{T-1} \in W^{T}$, such that 
\begin{itemize}
\item[1)] $x_0 = x$,  $x_T \in X \setminus X_{\rm safe}$, 
\item[2)] $\forall t\in \li 0, T-1\ri :  x_{t+1}= f(x_t, u_t, w_t), y_t = g(x_t, v_t)$, 
$u_t = \pi(\textbf{y}_t, \textbf{u}_{t-1})$, where $\textbf{u}_{-1} $ is defined to be the empty sequence $\varepsilon$. 
\end{itemize}
Since the sequences $\textbf{x}_T$, $\textbf{y}_{T-1}$ and $\textbf{u}_{T-1}$ can be determined from $x_0 = x$, $\textbf{v}_{T-1}$ and $\textbf{w}_{T-1}$ given $f$, $g$ and $\pi$, we will say that $\Sigma_\pi(x)$ is falsifiable under $(x, \textbf{w}_{T-1}, \textbf{v}_{T-1})$. 
The \textit{safety game winning set of a controller} $\pi$, denoted by $Win^\pi$, is defined to be $\{x \in X\mid \Sigma_\pi(x) \text{ is not falsifiable}\}$. 
\end{defn}

\subsection{A Motivating Example}

We now present a toy example to motivate the problem of finding ``controller-specific" adversarial scenarios under imperfect information. Consider the system in Figure. \ref{fig:eg_motivate}, 
where the circles represent states $x \in \li 1, 6\ri$, the blue and red arrows represent the transitions under two control actions $u \in \{blue, red\}$, and the numbers in quotation marks near a state $x$ represent the possible observations $y$ at state $x$.   The safe set $X_{\rm safe} = \li 1, 5\ri$. It can be verified that the maximal robust controlled invariant set $R = \{1,2,3, 4\}$ and the dual winning set $D = \{5,6\}$ with perfect information. The best state-feedback controller $\overline{\pi}_{\rm best}$ is such that $\overline{\pi}_{\rm best}(3) = blue$ and $\overline{\pi}_{\rm best}(4) = red$.

In the imperfect information case, let $\pi_1$ and $\pi_2$ be two output-feedback controllers. 
 Define $\pi_1$ to be such that $\pi_1(\textbf{y}_t, \textbf{u}_{t-1}) = blue$ if $y_0 \in \{ 1, 3, 4\} $, and $\pi_1(\textbf{y}_t, \textbf{u}_{t-1}) = red$ if $y_0 = 2$. 
 Define $\pi_2$ to be such that $\pi_2(\textbf{y}_t, \textbf{u}_{t-1}) = blue$ if $y_0 = 1$, and $\pi_2(\textbf{y}_t, \textbf{u}_{t-1}) = red$ if $y_0 \in \{2,3,4\}$. 
% Define $\pi_1$ to be such that $\pi_1(\textbf{y}_t, \textbf{u}_{t-1}) \equiv blue$  for all $(\textbf{y}_t, \textbf{u}_{t-1})$, 
% and define $\pi_2$ to be such that $\pi_2(\textbf{y}_t, \textbf{u}_{t-1}) $ $\equiv red$ for all $(\textbf{y}_t, \textbf{u}_{t-1})$. 
Clearly, $Win^{\pi_1} = \{1,2,3\}$ and $Win^{\pi_2} = \{1,2,4\}$, and these two winning sets are not comparable. 
More importantly, no other controller has a winning set larger than  $Win^{\pi_1}$ or  $Win^{\pi_2}$. 
\begin{figure}[h]
\centering
    \includegraphics[width=0.9\textwidth]{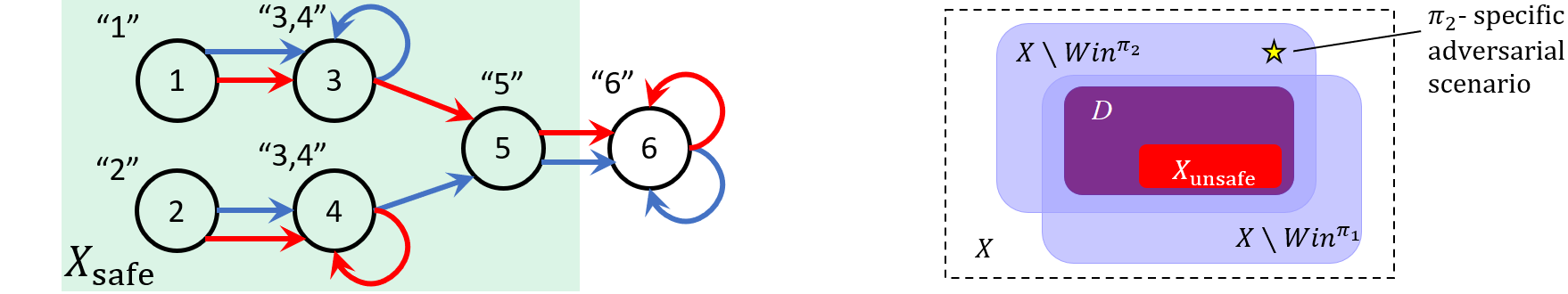}
      \caption{Left: motivating example. Right: illustration of controller-specific falsifications.}
    \label{fig:eg_motivate}
\end{figure}

By the above argument, any output-feedback controller $\pi$  is guaranteed to be falsified either starting from state $3$ or from state $4$. Note, however, that falsification at states $\{1, 2\}$ can be avoided by some controller (e.g., $\pi_1$ and $\pi_2$) although visiting state $3$ and state $4$ is inevitable from state $1$ and state $2$, respectively. 
This is because the two controllers have memory and can distinguish state $3$ from state $4$ if the system starts from either state $1$ or state $2$ and if the initial state is remembered. This is not the case, however, if we restrict the controllers to be static.    
This suggests that a backward expansion algorithm will not work unless the controllers to falsify are static.

The situation is illustrated by Figure \ref{fig:eg_motivate} (right). 
An adversarial scenario for controller $\pi_2$ can be found from the star, but can be avoided by controller $\pi_1$. Moreover, no controller is better than $\pi_2$ ($\pi_1$ is no better because its winning set is incomparable with that of $\pi_2$). 
We will refer to such falsifications as \textit{controller-specific adversarial scenarios}, 
to distinguish them from the generic adversarial scenarios obtained by solving the dual game under full information (see Remark \ref{rmk:generic}). 
\hfill\(\dagger\)

\subsection{Problem Statement}
Motivated by the above discussions, we  will focus on falsifying a closed-loop system running an output-feedback controller. 
Our goal is to find a controller-specific adversarial scenario under imperfect information. %To enable a backward expanding algorithm, we specialize to falsifying static controllers. 
The controllers under consideration in this work are static controllers. Such static controllers are common in feedback control (e.g., MPC, neural network controller). The formal problem statement  is given below.

\begin{prob}\label{prob:p1}
\normalfont
Consider a white-box system in Eq. \eqref{eq:syst}, and a black-box static output-feedback controller $\pi: Y \rightarrow U$ that can be queried at any given point $y\in Y$, (i.e., $u = \pi(y)$ can be obtained by querying). Given a set $X_{\rm init}$ of initial states, and a set $X_{\rm safe}$ of safe states, find an adversarial test case $(x,\textbf{w}_{T-1}, \textbf{v}_{T-1})$ such that 
$x \in X_{\rm init} \setminus D$, and $\Sigma_\pi(x)$ is falsifiable under $(x,\textbf{w}_{T-1}, \textbf{v}_{T-1})$.  
\end{prob}

In the above problem formulation, requiring $x \in X_{\rm init} \setminus D$ allows the information imperfection (i.e., sensor noise) to play an important role in the falsification (if $x \in D$, the adversarial disturbance by itself can lead to safety violation).  
The initial set $X_{\rm init}$ is usually far enough from the unsafe set so that trivial falsifying trajectories that start too ``close'' to the unsafe set do not exist.

\section{Synthesis-guided Adversarial Scenario Generation}\label{sec:sol}

We first discuss the challenge of solving Problem \ref{prob:p1}. 
Now that our goal is to find controller-specific falsifying scenarios, we must use the knowledge of the given controller $\pi$.  
The dual-game-based falsification approach does not apply directly because it only searches for generic adversarial scenarios, as discussed in Remark \ref{rmk:generic}. 
Theoretically, similar to Eq. \eqref{eq:iter11}, \eqref{eq:iter12}, one can search for controller-specific scenarios by solving a one-player game: 
\begin{align}
    X_0^\pi & = D,\label{eq:iter111} \\
    X_{k+1}^\pi & = \textbf{Pre}^{\pi}(X_k^{\pi}): = \big\{x \in X\, \big\vert \,\exists y \sim x, w \in W: f\big(x, \pi(y), w\big)\in X_{k}^\pi \big\}. \label{eq:iter112}
\end{align}
We stop iterating when $(X_N^\pi \setminus D) \cap X_{\rm init} \neq \emptyset$, and a falsifying trajectory can be extracted by  forward expansion. 
However, computing $\textbf{Pre}^{\pi}(\cdot)$ is possible only when the explicit expression of the given controller $\pi$ is fully known and also simple enough. 
In practice, it is more common that the given static output-feedback controller is complicated and even unknown. For example,  
this given controller can be a rule-based controller (i.e., look-up table) obtained by calibration or learning methods, an MPC controller, or even a black-box controller whose expression is unavailable to the test engineers. The backward reachable sets of the closed-loop systems under such controllers are difficult to compute in general. Instead, we are only allowed to simulate or query the given static controller at certain observation points, and this provides us with limited knowledge of the controller's behavior.  
The key challenge is to generate careful queries at critical points and avoid sampling the entire space exhaustively.

The rest of this section is devoted to tackling the above challenge. 
The key idea is to use synthesis, which does not require any knowledge of the given controller $\pi$, to guide where to query  $\pi$.  
 In particular, the controller is only queried at observation points carefully selected by solving a local dual game, so that exhaustive sampling of the output space is avoided. 
Based on the controller's decision at the selected points, we can under-approximate $\textbf{Pre}^{\pi}(\cdot)$. This leads to an  inner approximation of the backward reachable set $X_k^\pi$, which would be sufficient for the purpose of falsification because we only aim at finding one falsifying trajectory.

In what follows, we will introduce different ingredients in this procedure and conclude the section with a pseudocode that integrate all these ingredients.

\subsection{Under-approximate $\textbf{Pre}^{\pi}(\cdot)$ with Local Dual Game and Query}
First, the inner approximation of $X_k^\pi$  can be constructed by an iterative process that alternatingly solves a local dual game under imperfect information \eqref{eq:iter21} and queries the given controller \eqref{eq:iter23}.
\begin{align}
    \underline{X}_0^\pi & \subseteq D, \label{eq:iter21}\\
    \hspace{-2mm}\underline{Y}_{k+1} & = \textbf{EPre}_y(\underline{X}_k^{\pi} \mid U) := \{y \in Y\,\vert\,\forall u \in U: \exists x\sim y, w \in W: f(x, u, w) \in \underline{X}_k^\pi \}, \label{eq:iter22} \\
    \underline{y}_{k+1} & \in \underline{Y}_{k+1}, \label{eq:iter23}\\
    \underline{X}_{k+1}^\pi & =  \textbf{Pre}_{\underline{y}_{k+1}}^\pi(\underline{X}_k^{\pi}) := \big\{x \in X  \big\vert x\sim \underline{y}_{k+1}, \exists w \in W: f\big(x, \pi(\underline{y}_{k+1}), w\big)\in \underline{X}_k^\pi \big\}. \label{eq:iter24}
\end{align}
In Eq. \eqref{eq:iter21}-\eqref{eq:iter24}, we start from the dual game winning set $D$ with perfect information and expand backwards. Set $\underline{Y}_{k+1}$ is computed as the solution of a local dual game with imperfect information,  consisting of the output points $y$ for which no matter what control input $u$ is picked at $y$, a violation will occur starting from some $x\sim y$ ($x$ depends on $u$). 
The operation $\textbf{EPre}_y(\cdot)$ is similar to the environment-controller predecessor but in the observation space (indicated by subscript $y$). 
In the motivating example, $\underline{X}_0 = \{5, 6\}$ and $\underline{Y}_{1} = \{3, 4\}$. 
The significance of step \eqref{eq:iter22} is that it does not require any knowledge of $\pi$ to compute $\textbf{EPre}_y(\cdot)$, and hence the computational complexity is independent of the complexity of $\pi$. We then pick an output point $\underline{y}_{k+1}$ from $\underline{Y}_{k+1}$ and simulate the closed-loop system under $\pi(\underline{y}_{k+1})$. This allows us to compute $\underline{X}_{k+1}^\pi$, the set of states from which visiting $\underline{X}_k^\pi$ can be enforced by the environment. Note that, to compute the set $\underline{X}_{k+1}^\pi$, we only need to know the decision of $\pi$ at one point $\underline{y}_{k+1}$. This makes sure that, again, the complexity of the computation does not depend on the complexity of  $\pi$. 

As promised at the beginning of this section, for the alternating iteration process to be correct, the set  $\underline{X}_k^\pi$ should be an inner approximation of $X_k^\pi$. This is stated and proved below.

\begin{thm}\label{thm:main}
\normalfont
Let $\underline{X}_k^\pi$ and $X_k^\pi$ be the sets computed by iterations \eqref{eq:iter21}-\eqref{eq:iter24}, 
%  \eqref{eq:iter21}-\eqref{eq:iter24} 
and iterations \eqref{eq:iter111}, \eqref{eq:iter112}  respectively, then 
$\underline{X}_k^\pi \subseteq X_k^\pi$.
\end{thm}

\begin{proof}
We prove this by induction. 
\begin{enumerate}% [nolistsep]
    \item[$1^\circ$] Base case: by \eqref{eq:iter111} and \eqref{eq:iter21}, $\underline{X}_0^\pi \subseteq D  = X_0^\pi$. 
    \item[$2^\circ$] 
% Induction step: 
Assume $\underline{X}_k^\pi \subseteq X_k^\pi$, we prove $\underline{X}_{k+1}^\pi \subseteq X_{k+1}^\pi$. Define \vspace{-1.5mm}
    \begin{align}
        XY_{k+1}^\pi & : = \big\{(x,y) \in X\times Y\mid x\sim y, \exists w \in W: f\big(x,\pi(y), w\big)\in X_k^\pi \big\}, \label{eq:XYpi}\\
        Y_{k+1}^\pi & : = \big\{y\in Y \mid \exists x\sim y, w\in W: f\big(x, \pi(y), w\big)\in  X_k^\pi \big\}, \\
        \widetilde{X}_{k+1}^\pi & : =  \big\{x \in X \, \, \big\vert\, x\sim \underline{y}_{k+1}, \exists w \in W: f\big(x, \pi(\underline{y}_{k+1}), w\big) \in X_k^\pi\big\}. \label{eq:Xpi}
    \end{align}
    By definition, $X_{k+1}^\pi$ is the projection of $XY_{k+1}^\pi$ onto $X$ space, $Y_{k+1}^\pi$ is the projection of $XY_{k+1}^\pi$ onto $Y$ space, and $\widetilde{X}_{k+1}^\pi$ is the slice of $XY_{k+1}^\pi$ at $y = \underline{y}_{k+1}$ (see Fig. \ref{fig:prf}). 
    Clearly, we have 
       $ \underline{y}_{k+1} \in \underline{Y}_{k+1} \subseteq Y_{k+1}^\pi$, 
    and this implies $\widetilde{X}_{k+1}^\pi \subseteq X_{k+1}^\pi$. 
 Note that  we also have $\underline{X}_{k+1}^\pi\subseteq\widetilde{X}_{k+1}^\pi  $ by the induction hypothesis $\underline{X}_k^\pi \subseteq X_k^\pi$. Thus
       $ \underline{X}_{k+1}^\pi \subseteq \widetilde{X}_{k+1}^\pi \subseteq X_{k+1}^\pi$,  
    which completes the induction step and the entire proof. 
\end{enumerate}
\vspace{-4.5mm} 
\end{proof}

Once the sets $\underline{X}_{k}^\pi$ and the observations $\underline{y}_{k+1}$ are computed, an adversarial scenario $(x_0, \textbf{w}_{N-1}, \textbf{v}_{N-1})$ can be derived using a forward expansion procedure:  
\begin{align}
    x_0 & \in \underline{X}_N^\pi \cap X_{\rm init}, \label{eq:a_fwd_e1}\\
    y_t & =\underline{y}_{N-t}, \ \  v_t \text{ is s.t. } g(x_t, v_t) = y_t \text{ for } t\in \li 0, N-1\ri, \\
    w_t & \text{ is s.t. }  f\big(x_t, \pi(y_t), w_t\big) \in \underline{X}_{N-t - 1}^\pi, \ \  x_{t+1} = f\big(x_t, \pi(y_t), w_t\big) \text{ for } t\in \li 0, N-1\ri.  \label{eq:a_fwd_e3}
\end{align}
Clearly, $x_N \in \underline{X}_0^\pi \subseteq D$. The full scenario $(x_0, \textbf{w}_{T-1}, \textbf{v}_{T-1})$ can be completed by Eq. \eqref{eq:d_fwd_e1}, \eqref{eq:d_fwd_e2} so that $x_T \in X_{\rm unsafe}$, with $v_N, \dots, v_T$ being arbitrarily picked.

\begin{figure}[hbt]
\centering
    \includegraphics[width=0.4\textwidth]{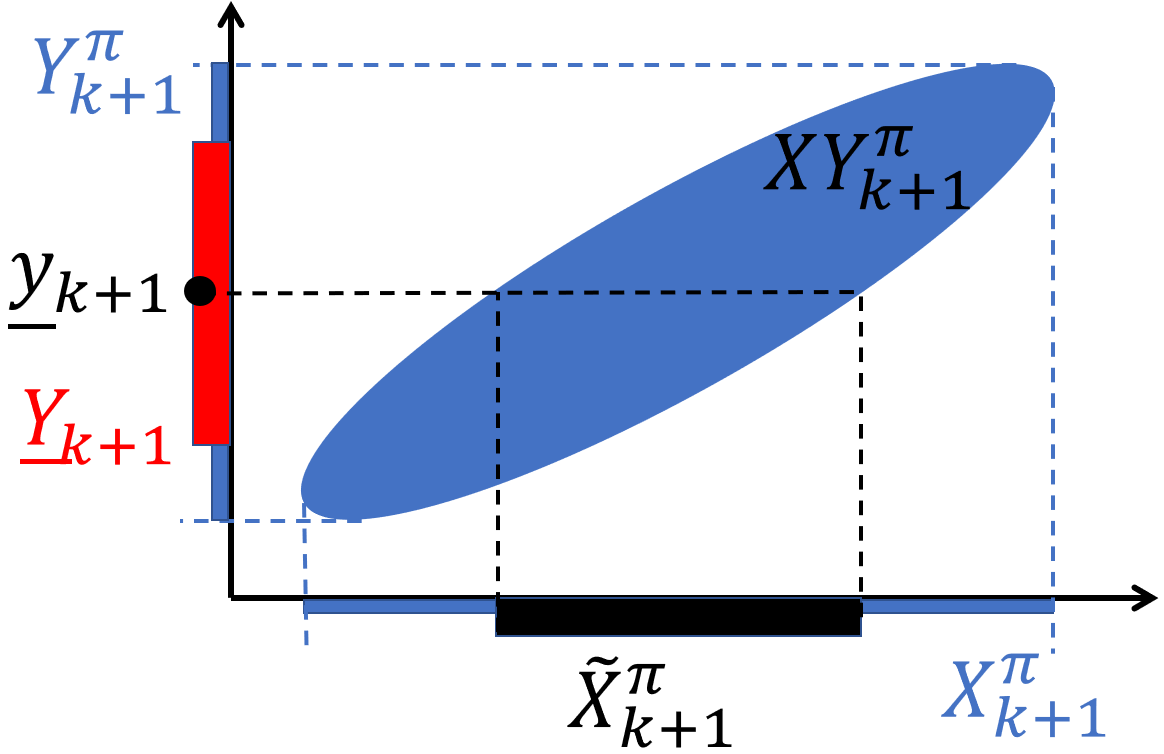}
    \caption{Illustration: proof of Theorem 1.}
    \label{fig:prf}
% \end{wrapfigure}
\end{figure}
Note that the proposed algorithm is not complete in that it does not guarantee to find a falsification scenario even if one exists. 
Otherwise it will be equivalent to verification in the sense that a system is verified whenever no falsifying scenario is found. 
Such a complete procedure, however, is computationally hard (e.g., most verification problems with continuous state-spaces are undecidable \cite{henzinger2000robust}) and is not the point of the falsification problem. Moreover, it is not clear how to do verification with unknown components without resorting to statistical techniques. 
% \end{rmk}

\subsection{Computing $\underline{X}_k^\pi$ and $\underline{y}_k$ for Switched Affine Systems }
As pointed out earlier, the complexity of the set computation in the alternating iteration \eqref{eq:iter21}-\eqref{eq:iter24} 
% only depends on the form of the open-loop system, and
 is independent of the controller's complexity. Hence it is applicable to systems whose open-loop dynamical model $f$ and the sensor model $g$ have modest complexity. 
For such systems, sets $\underline{X}_k^\pi$, $\underline{Y}_k$ can be described and computed efficiently using e.g., HJB methods \cite{mitchell2005time}, optimization-based approaches \cite{lasserre2015tractable}, interval analysis \cite{li2017invariance}, 
just to mention a few. In what follows,  we show, for switched affine systems with linear constraints, how to do the set computations by leveraging the existing polytope/zonotope computation techniques. Our implementation includes both of these options.
While the polytope-based implementation is faster for 2D and 3D plant models, the zonotope-based implementation scales better with the dimension of the state space.

\begin{prop}\label{prop:swaP}
\normalfont
Let $\Sigma$ be the following switched affine system, 
\begin{align}
x_{t+1} & = A_{s_t} x_t + B_{s_t} u_t + K_{s_t} + E_{s_t}w_t,  \label{eq:syst_swa} \\
y_t & = C x_t + Du_t + Fv_t, \label{eq:sysm_swa}
\end{align}
where $s_t$ is a switching control input from a finite set $S$. Suppose that set $\underline{X}_k^\pi$ and the continuous control input set $U$ are polyhedrons, i.e., $\underline{X}_k^\pi = \{x\in \mathbb{R}^{n_x}\mid Hx \leq h\}$, and $U = \text{cvxh}(\mathcal{V}_u)$ where $\text{cvxh}$ denotes the convex hull and $\mathcal{V}_u \subseteq \mathbb{R}^{n_u}$ is a finite set of vertices. 
Also assume that, at a given state $x$, the (potentially state-dependent) process noise $w \in W(x) := \{w\in \mathbb{R}^{n_w}\mid H_x x + H_w w \leq h_w\}$, and the sensor noise $v \in V(x) := \{v\in \mathbb{R}^{n_v}\mid G_x x + G_v v \leq g_v\}$.  
Then $\underline{y}_{k+1} \in \underline{Y}_{k+1}$ can be found by solving the linear program in  \eqref{eq:YkP}, and the set  $\underline{X}_{k+1}^\pi$ is the polyhedron in \eqref{eq:XkP}. 
\begin{align}
&  \begin{array}{rl}
\text{ find } & \underline{y}_{k+1}  \in \mathbb{R}^{n_y}, \ \ \{x_{s,u}, \ w_{s,u}, \ v_{s,u}\}_{s\in S, u\in \mathcal{V}_u} \\
 \text{ s.t. } & \forall s\in S, u\in \mathcal{V}_u: \\
& H(A_s x_{s,u} + B_s u + K_s + E_s w_{s,u}) \leq h, \ \ C x_{s,u} + Du + Fv_{s,u} = \underline{y}_{k+1}, \\
& H_x x_{s,u} + H_w w_{s,u} \leq h_w, \ \ G_x x_{s,u} + G_v v_{s,u} \leq g_v 
\end{array}, 
 \label{eq:YkP}
\\
& \underline{X}_{k+1}^{\pi}\hspace{-0mm}=\hspace{-0mm}\textbf{Proj}_x\hspace{-0mm}\left(\hspace{-0mm}\left\{\hspace{-0mm}\left.
\begin{array}{l} x  \in \mathbb{R}^{n_x} \\
w \in \mathbb{R}^{n_w}  \\
v \in \mathbb{R}^{n_v}  \end{array}
\hspace{-1mm}\right\vert \hspace{-0mm}
\begin{array}{l}
H(A_{s_{k+1}} x + B_{s_{k+1}} u + K_{s_{k+1}} + E_{s_{k+1}} w) \leq h \\
C x + Du_{k+1} + Fv = \underline{y}_{k+1}\\
H_x x+ H_w w \leq h_w, \ G_x x + G_v v \leq g_v\\
\end{array}
\right\}\right).
 \label{eq:XkP}
\end{align}
In \eqref{eq:XkP}, $(s_{k+1}, u_{k+1}) = \pi(\underline{y}_{k+1})$ and $\textbf{Proj}_x(\mathcal{X}) = \{x \mid \exists z: (x,z)\in \mathcal{X}\}$ is the polytope projection from some high-dimensional polytope $\mathcal{X}$. 
\end{prop}

We implement \eqref{eq:YkP}, \eqref{eq:XkP} with MPT3 toolbox \cite{MPT3}. 
Due to the complexity of the projection operation in \eqref{eq:XkP}, we are restricted to plants with low-dimensional state spaces (i.e., $n_x \leq 4$). 
To handle the cases where $n_x > 4$, we use zonotopes to under-approximate the polytopic set  $\underline{X}_{k+1}^\pi$.

\begin{prop}\label{prop:zono}
\normalfont 
Consider a switched affine system $\Sigma$ in \eqref{eq:syst_swa}, \eqref{eq:sysm_swa}. 
Assume that $C = I$, $A_s$ is invertible\footnote{This assumption holds when the discrete-time model is obtained by sampling a continuous-time linear system. } for all $s\in S$, and that $U = \langle G_u, c_u\rangle$, $W = \langle G_w, c_w\rangle$, $V = \langle G_v, c_v\rangle$ and $\underline{X}_k^\pi = \langle G_x, c_x\rangle$ are zonotopes\footnote{$\langle G,c\rangle$ is the generator-representation of the zonotope $\{G\theta + c \mid \theta \in [-1,1]^N\}$, where $G \in \mathbb{R}^{n\times N}$, $c\in \mathbb{R}^n$.}. 
Then $\underline{y}_{k+1}$ is a feasible solution of the linear program in \eqref{eq:YkZ}, and $\underline{X}_{k+1}^\pi$ is the intersection of two zonotopes in \eqref{eq:XkZ}.   
\begin{align}
&  \begin{array}{rl}
\text{ find } & \underline{y}_{k+1}  \in \mathbb{R}^{n_y}, \ \ \{x_{s,u}, \ \xi_{s,u}, \ w_{s,u}, \ \omega_{s,u}, \ v_{s,u},  \nu_{s,u}, \ \}_{s\in S, u\in \mathcal{V}_u} \\
 \text{ s.t. } & \forall s\in S, u\in \mathcal{V}_u: \\
& A_s x_{s,u} + B_s u + K_s + E_s w_{s,u} = G_x \xi_{s,u} + c_x, \ \ C x_{s,u} + Du + Fv_{s,u} = \underline{y}_{k+1}, \\
& w_{s,u} = G_w \omega_{s,u} + c_w, \ \ v_{s,u} = G_v \nu_{s,u} + c_v
\end{array}, 
 \label{eq:YkZ}
\\
& \underline{X}_{k+1}^\pi = A_{s_{k+1}}^{-1}(\underline{X}_k^\pi \oplus -E_{s_{k+1}} W - B_{s_{k+1}} u_{k+1} - K_{s_{k+1}}) \cap (\underline{y}_{k+1} \oplus -FV - Du_{k+1}). 
\label{eq:XkZ}
\end{align}
In \eqref{eq:XkZ},  $(s_{k+1}, u_{k+1}) = \pi(\underline{y}_{k+1})$, $\langle G_1, c_1\rangle \oplus \langle G_1, c_1\rangle := \langle [G_1, G_2], c_1+c_2\rangle$ and $H\langle G, c\rangle := \langle HG, Hc\rangle$. 
\end{prop}

The main challenge in computing the set $\underline{X}_{k+1}^\pi$ using Proposition \ref{prop:zono} is the intersection operation in \eqref{eq:XkZ}. 
In fact, the intersection of two zonotopes is not necessarily a zonotope. However, we can find a zonotopic under-approximation of the intersection by solving a linear program. That is, the intersection of two zonotopes $\langle G_1, c_1\rangle$, $\langle G_2, c_2\rangle$ can be under-estimated by $\langle \Lambda [G_1, G_2], c\rangle$, where the diagonal matrix $\Lambda = \textbf{diag}(\{\lambda_i\}_i)$ is the solution of the following optimization problem: 
\begin{align}
\begin{array}{rl} 
\max_{\lambda} & \sum_i \lambda_i \\
\text{s.t. } & \langle \Lambda [G_1, G_2], c\rangle \subseteq \langle G_1, c_1\rangle \\
& \langle \Lambda [G_1, G_2], c\rangle \subseteq \langle G_2, c_2\rangle \\
& \lambda_i \geq 0, \forall i 
\end{array}. 
\label{eq:zc}
\end{align}
The zonotope containment constraints in \eqref{eq:zc} can be reformulated into a set of linear constraints using the technique  developed in \cite{sadraddini2019linear}. Therefore \eqref{eq:zc} is equivalent to a linear program. 
The obtained zonotope $\langle \Lambda [G_1, G_2], c\rangle$ is an inner approximation of $\underline{X}_{k+1}^\pi$. Moreover, for low dimensional systems (i.e., $n_x = 2$ or $3$), solving the above optimization problem is slower than computing the polytopic $\underline{X}_{k+1}^\pi$ as in Proposition \ref{prop:swaP}. However, the zonotope implementation scales better with $n_x$ and can easily go beyond what can be solve via the off-the-shelf polytope computation tools. For example, the zonotope-based implementation can handle systems with $n_x = 10$.

\subsection{Improving Completeness}\label{sec:complete}

We provide a heuristic and a refinement technique to improve the completeness of the proposed alternating backward expansion \eqref{eq:iter21}-\eqref{eq:iter24}.

\textbf{Selection of the queried point $\underline{y}_{k+1}$}. 
The completeness of our approach can be improved by a better selection of the observation point $\underline{y}_{k+1}$ in Eq. \eqref{eq:iter22}. 
One natural heuristic is to pick $\underline{y}_{k+1} \in \underline{Y}_{k+1}$ that minimizes the distance from the set $\{x\mid x\sim\underline{y}_{k+1}\}$ to the initial set $X_{\rm init}$. However, this usually leads to set $\underline{X}_{k+1}^\pi$ having zero measure and $\underline{Y}_{k+2}= \emptyset$.  
If such situation occurs before $\underline{X}_k^\pi$ reaches  $X_{\rm init}$, we may end up with an unsatisfactory falsifying trajectory that starts too close to the unsafe set, 
and the alternating backward expansion is said incomplete if falsifying trajectories starting from 
 $X_{\rm init}$ exist but the algorithm does not find one due to a bad selection of $\underline{y}_{k+1}$. 
Hence there is a trade-off between picking a $\underline{y}_{k+1}$ that leads to a state closer to the initial set $X_{\rm init}$, and picking a $\underline{y}_{k+1}$ that leads to a larger set $\underline{X}_{k+1}^\pi$.  Here we use the following  simple heuristic 
\vspace{-0mm}
 \begin{align}
\underline{y}_{k+1} = \beta y_{\rm close} + (1 - \beta) y_{\rm center}, 
\label{eq:beta}
\end{align}
where $y_{\rm close}$ minimizes the Euclidean distance between set $\{x\mid x\sim y_{\rm close}\} $ and  
% the initial set
 $X_{\rm init}$, $y_{\rm center}$ is an approximated center of the polytopic set $\underline{Y}_{k+1}$, and $\beta \in [0,1]$ is a tuning factor. 
In this paper, we use $\beta = 0.6$ for all continuous-state systems. % the examples with continuous state spaces.
% In Example \ref{eg:cont}, the sets $\underline{X}_k^\pi$, $\underline{Y}_k$ shrinks very quickly as $k$ increases.
% \vspace{-0.5mm}

\textbf{Refining the control set $U$}. 
Another reason for set $\underline{Y}_{k+1}$ becoming empty before the falsifying trajectory reaches $X_{\rm init}$ during backward expansion is the control input set $U$ being too large. 
In this case, it is hard for an output $y$ to satisfy the condition in the definition of $\underline{Y}_{k+1}$ (see Eq.  \eqref{eq:iter22}) for all $u \in U$. 
However, for the alternating backward expansion to be valid, it is sufficient for $y$ to satisfy the condition only for $u = \pi(y)$ (though this requires the knowledge of $\pi$ and is undesired). 
One way to mitigate this problem is to refine the control input space $U$ into finitely many smaller sets $\{U_i\}_{i=1}^n$ and replace 
$\underline{Y}_{k+1}$ by $\underline{Y}_{k+1, i} : = \textbf{EPre}_y(\underline{X}_k^\pi \mid U_i) = \{y\in Y \mid \forall u\in U_i: \exists x \sim y, w\in W: f\big(x, \pi(y), w\big) \in \underline{X}_{k}^\pi \}$.  
The following proposition can be easily proved under some continuity assumptions and suggests that a nonempty $\underline{Y}_{k+1, i}$ can be always found if the partition of $U$ is fine enough.  
% Here we just point out this possibility of improving the algorithm's completeness when $U$ is large. 
% In the examples in this paper, this refinement was not needed and is not used. 
%\end{rmk}

\begin{prop}\label{prop:refine}
\normalfont
Suppose that $X, U$ are normed spaces and the mapping $f$ is continuous in $u$, and that there exists $y\in Y$, $x\sim y$ and $w \in W$ such that $f\big(x, \pi(y), w\big)$ is in the interior of $\underline{X}_{k}^\pi$. There exists $\delta > 0$ such that any $\delta$-fine partition $\{U_i\}_{i=1}^n$ of set $U$ (i.e., $U = \bigcup_{i=1}^n U_i$ and for any $u, u' \in U_i: \Vert u - u'\Vert\leq \delta$) contains a piece $U_{i'}$ s.t. $\underline{Y}_{k+1, i'} = \textbf{EPre}_y(\underline{X}_k^\pi \mid U_{i'})  \neq \emptyset$. 
%  = \{y \in Y \mid \forall u\in \widehat{U}_{i'}: \exists x \sim y: \tau(x,u)\cap \underline{X}_k^\pi \neq \emptyset\} \neq \emptyset$.
\end{prop}

\begin{proof}
Let $y \in Y$, $x \sim y$ and $w \in W$ such that $f\big(x,\pi(y), w\big)$ is in the interior of $\underline{X}_{k}^\pi$, i.e., $\{x' \mid \Vert x' - f(x,\pi(y),w)\Vert \leq \epsilon\} \subseteq \underline{X}_{k}^\pi$ for some $\epsilon > 0$. Since $f$ is continuous in $u$, there exists $\delta$ s.t. $\Vert u - \pi(y)\Vert \leq \delta  \Rightarrow \Vert f(x,u,w) - f\big(x, \pi(y), w\big)\Vert \leq \epsilon$. Hence for all $u$ such that $\Vert u - \pi(y)\Vert \leq \delta$, $f(x,u,w) \in \underline{X}_k^\pi$. Let $\{U_i\}_{i=1}^n$ be any $\delta$-fine partition of set $U$ and let $\pi(y) \in U_{i'}$. One have $\Vert u - \pi(y)\Vert \leq \delta$ for all $u \in U_{i'}$. 
This implies that $y \in \underline{Y}_{k+1, i'}$ and thus $\underline{Y}_{k+1, i'} \neq \emptyset$. 
\end{proof}

\subsection{Overall Algorithm}
Algorithm \ref{alg:main} summarizes our falsification procedure, integrating the alternating iteration process, the query heuristic and the control space refinement.

\begin{algorithm}[h]
\caption{$(x_0, \textbf{w}_{T-1}, \textbf{v}_{T-1}) = {\textbf{FindAdversarialScenario}}(\Sigma, \pi,  X_{\rm init},  X_{\rm safe}, k_{\rm max}, \liren{\delta})$}
\begin{algorithmic}[1]
  \STATE Perform iteration \eqref{eq:bkwd_e1}, \eqref{eq:bkwd_e2} to get $\{\underline{D}_k\}_{k=0}^K$, where $\underline{D}_k \subseteq D_k$ ($\underline{D}_k \neq D_k$ because e.g., the union operation in Eq. \eqref{eq:bkwd_e2} is dropped to simplify the computation)
   \STATE Pick $\underline{X}_{0}^\pi \subseteq \bigcup_{k=1}^K \underline{D}_k$, set $k=0$
\WHILE{$ \underline{X}_{k}^\pi \cap X_{\rm init} = \emptyset$ \text{and} $k \leq k_{\rm max}$}
\STATE  \liren{$\mathcal{U} = \{U\}$} 
 \STATE \textbf{do} pick $U' \in \mathcal{U}$ , and  $\underline{y}_{k+1} \in  \textbf{EPre}_y(\underline{X}_k^\pi\mid U')$  using the heuristic in Eq. \eqref{eq:beta}
 \STATE  \ \ \ \ \textbf{if} $\exists u, u' \in U': \Vert u - u'\Vert > \delta$
\STATE \ \ \ \ \ \ \ split $U'$ into $U_1$ and $U_2$, $\mathcal{U}= (\mathcal{U} \setminus \{U'\})\cap \{U_1, U_2\}$
\STATE \ \ \ \ \textbf{endif}
\STATE \textbf{while} ($\underline{y}_{k+1}$ does not exist \textbf{or} $\pi(\underline{y}_{k+1})\notin U'$) \textbf{and} $\mathcal{U}\neq \emptyset$
\STATE $\underline{X}_{k+1}^\pi =  \textbf{Pre}^\pi_{\underline{y}_{k+1}}( \underline{X}_k^\pi)$ 
\STATE update $k = k+1$
\ENDWHILE
\STATE Set $N = k$,  generate $\textbf{x}_{N}, \textbf{w}_{N-1}, \textbf{v}_{N-1}$ using Eq. \eqref{eq:a_fwd_e1}-\eqref{eq:a_fwd_e3} 
\STATE From $x_N$, complete $\textbf{x}_{T}, \textbf{w}_{T-1}, \textbf{v}_{T-1}$ using Eq. \eqref{eq:d_fwd_e1}, \eqref{eq:d_fwd_e2}, with ``$\overline{\pi}(x_t)$'' replaced by $\pi\big(g(x_t, v_t)\big)$ where $v_N, \dots, v_{T-1}$ are picked randomly. 
\RETURN  $x_0, \textbf{w}_{T-1}, \textbf{v}_{T-1}$
\end{algorithmic}\label{alg:main}
\end{algorithm}

We briefly explain Algorithm \ref{alg:main}. 
Line 4-9 corresponds to the refinement technique introduced in Section \ref{sec:complete}. 
While in some applications, the control set $U$ may be physically constrained and it is possible to find $\underline{y}_{k+1} \in  \textbf{EPre}_y(\underline{X}_k^\pi\mid U)$, the refinement is important when $U$ is large (or even unconstrained). 
Algorithm \ref{alg:main} handles the latter case by splitting the control set $U$, saving different pieces in a queue $\mathcal{U}$ and finding $\underline{y}_{k+1} \in  \textbf{EPre}_y(\underline{X}_k^\pi\mid U')$ for each $U' \in \mathcal{U}$. 
By Proposition \ref{prop:refine}, under some continuity assumptions on $f$, it is possible to find $\underline{y}_{k+1}$ when the partition of $U$ is fine enough (quantified by the input parameter $\delta$). 
An example where the refinement is necessary to find a satisfying adversarial scenario (i.e., $x_0 \in X_{\rm init}$) can be found in Section \ref{sec:eval} (Example \ref{eg:lqr_unknown}). Note that, after splitting $U$, there might be multiple $\underline{y}_{k+1}$ induced from different parts of $U$, from where the expansion can proceed. 
In that case, it is possible to introduce a backtracking mechanism into our framework to manage the branching.  
This increases the computational complexity but can improve the completeness of the algorithm, especially when the ``greedy'' heuristic in \eqref{eq:beta} does not lead the search towards initial set $X_{\rm init}$ globally. 

In Algorithm \ref{alg:main}, the input parameters $k_{\rm max}$, $\delta$ can be adapted to different problems. 
In our examples, $k^{\rm max}$ is never reached and $\delta$ is only used in Example \ref{eg:lqr_unknown}. 
There are also several places where different heuristics can be used. 
One is the choice of $\underline{X}_0^\pi$ where the backward expansion starts. The set $\underline{X}_0^\pi$ should be large enough and closer to the initial set. In all of our examples, we pick $\underline{X}_0^\pi = \underline{D}_K$. However, in most of our examples, we find comparable results by picking $\underline{X}_0^\pi = X_{\rm unsafe}$. One can also adopt different heuristics for picking multiple $\underline{y}_{k+1}$'s. This will lead to branching in the backward expansion and the backtracking like ideas can be used to trace different branches.

\section{Extensions}\label{sec:ext}

The main development so far is an approach to falsify systems with relatively simple open-loop dynamics $f$, sensor model $g$, but with potentially complicated or even unknown static controller $\pi$, against invariance requirements. 
Here, we further show that our approach is extensible to
 i) systems whose dynamic model $f$ is complex but can be approximated by an abstraction with a simpler expression, 
ii) vision-based control systems, whose dynamics model $f$ is also simple but the sensor model $g$ and the controller $\pi$ are both complicated/unknown; and iii) against certain safety requirements expressed in a fragment of temporal logic. 
Each extension is illustrated with an example in Section \ref{sec:eval}.

\subsection{Incorporating Model Mismatch}\label{sec:ext_mismatch}
Since a simple plant model is assumed in our setting, it is important to incorporate the mismatch between our simple model and a high fidelity model, or even the real system. Our goal is to generate an adversarial scenario for the real system by falsifying the simple model, which serves as an abstraction of the real system. To this end, we consider simple models parametrized by model uncertainties $p$ and $q$ (their distinction will be discussed later). 
These model uncertainties $p$, $q$ are different from the disturbance $w$ and the noise $v$ that are picked adversarially in the falsification. Instead, $p$ and $q$ need to be picked in a way to respect the physics. 
This suggests another player called the neutral player selecting $p$, $q$ against the adversarial environment and independently of the controller. 
Let $x_{t+1} = f(x_t, u_t, w_t)$ be the concrete dynamics and let $x_{t+1} = \widehat{f}_{p_t, q_t}(x_t, u_t,w_t)$, where $p_t \in P$ and $q_t \in Q$, be an abstraction that has a simpler expression than $f$. 
The variable $p_t$ captures the terms in $f$ that can be determined after $x_t$ and $u_t$ are given, and these terms are known but hard to incorporate in backward reachable set computation.
Whereas $q_t$ captures the terms in $f$ that can be determined only after $w_t$ is given. 
For example, $f(x_t,u_t,w_t) = x_t + u_t + w_t + \log(x_t) + u_t^2 + u_t\cos(w_t)$ and $\widehat{f}_{p_t,q_t}(x_t,u_t,w_t) = x_t + w_t + u_t + p_t + q_t$ with $p_t \in P: = \{\log(x) + u^2\mid x \in X, u \in U\}$ and $q_t \in Q:=\{u\cos(w)\mid w \in W, u\in U\}$. 
To construct an adversarial scenario for the concrete system by falsifying the abstraction, we require the following condition to hold for $f$ and $\widehat{f}_{p,q}$: 
\begin{align}
\forall x\in X, u\in U: \exists p \in P: \forall w \in W:  \exists q \in Q: f(x,u,w) = \widehat{f}_{p,q}(x,u,w). 
\label{eq:relation}
\end{align}
The reason for treating $p$ and $q$ differently roots in the order of play in the game, i.e., the adversarial environment plays after the controller. When the adversarial environment is choosing $w$, the value of $p$ is already determined by $x$ and $u$, whereas the value of $q$ remains undetermined. Therefore the choice of $w$ may depend on $p$ but not $q$. For instance, $p$ can represent known nonlinearities that are abstracted out in backward reachability but whose value can be easily computed given $x$ and $u$, whereas $q$ represents unknown but bounded model uncertainties.
In each round of the game, the state is at $x_t$ and the adversarial environment picks $v_t$ that leads to $y_t = g(x_t, v_t)$ and $u_t = \pi(y_t)$. Next the neutral player picks $p_t$ so that 
\begin{align}
\forall w \in W:  \exists q \in Q: f(x_t,u_t,w) = \widehat{f}_{p_t,q}(x_t,u_t,w). 
\end{align} 
We know such $p_t$ exists by \eqref{eq:relation}. In the above example, the neutral player picks $p_t = \log(x_t) + u_t^2$. 
Then the adversarial environment picks $w_t$ according to the value of $x_t$, $u_t$ and $p_t$.  
Finally the neutral player picks $q_t$ so that $f(x_t, u_t, w_t) = \widehat{f}_{p_t, q_t}(x_t, u_t, w_t)$. In the above example, the neutral player picks $q_t = u_t \cos(w_t)$. 
The alternating iteration is modified by replacing \eqref{eq:iter22}, \eqref{eq:iter24}  by 
\begin{align}
    \underline{Y}_{k+1} & = \{y \in Y\,\vert\,\forall u \in U: \exists x\sim y: \forall  p\in P: \exists w \in W: \forall q \in Q: \widehat{f}_{p,q}(x,u,w) \in  \underline{X}_k^\pi \}, \label{eq:iter32} \\
    \underline{X}_{k+1}^\pi & = \big\{x \in X  \big\vert x\sim \underline{y}_{k+1}, \forall p \in P: \exists w \in W: \forall q \in Q: \widehat{f}_{p,q}\big(x, \pi(\underline{y}_{k+1}),w\big)\in \underline{X}_k^\pi \big\}. \label{eq:iter34}
\end{align}
The following proposition says that an adversarial scenario for the abstraction also falsifies the concrete system.  

\begin{prop}\label{prop:model_mismatch}
Replace \eqref{eq:iter22} by \eqref{eq:iter32}, and \eqref{eq:iter24} by \eqref{eq:iter34}, and let $X_k^\pi$ be defined by \eqref{eq:iter112} with the concrete dynamics. We have $\underline{X}_{k}^\pi \subseteq X_k^\pi$.
\end{prop}

To prove Proposition \ref{prop:model_mismatch}, we first introduce a useful lemma. 
\begin{lem}\label{lem:A1}
\normalfont 
Let $x' \in X$, $u' \in U$ and $X' \subseteq X$ be arbitrary. Suppose that \eqref{eq:relation} holds, i.e., 
\begin{align}
\exists p' \in P: \forall w \in W: \exists q \in Q: f(x',u',w) = \widehat{f}_{p',q}(x',u',w), \label{eq:A1}
\end{align}
and that 
\begin{align}
\forall p \in P: \exists w \in W: \forall q \in Q: \widehat{f}_{p,q}(x',u',w) \in X'. \label{eq:A2}
\end{align}
Then there exists $w' \in W: f(x',u',w') \in X'$. 
\end{lem}

\begin{proof}
Let $p = p'$  in \eqref{eq:A2}, we have 
\begin{align}
\exists w' \in W: \forall q \in Q: \widehat{f}_{p',q}(x',u',w') \in X'. \label{eq:A3}
\end{align}
Combining \eqref{eq:A1} with \eqref{eq:A3} yields
\begin{align}
\exists q' \in Q: f(x',u',w') = \widehat{f}_{p',q'}(x',u',w'). \label{eq:A4}
\end{align}
Eq. \eqref{eq:A3} gives  $\widehat{f}_{p',q'}(x',u',w') \in X'$, which implies that $f(x',u',w') \in X'$ by \eqref{eq:A4}.  
\end{proof}

Proposition \ref{prop:model_mismatch} is proved below. 
\begin{proof} We will define $XY_{k+1}^\pi$, $Y_{k+1}^\pi$ and $\widetilde{X}_{k+1}^\pi$ by \eqref{eq:XYpi}-\eqref{eq:Xpi}. It remains to show that, given $\underline{X}_{k}^\pi \subseteq X_k^\pi$,  i) the newly defined $\underline{Y}_{k+1} \subseteq Y_{k+1}^\pi$ and ii) the newly defined $\underline{X}_{k+1}^\pi \subseteq \widetilde{X}_{k+1}^\pi$. 
Then Proposition \ref{prop:model_mismatch} follows from the same induction that proves Theorem 1. 
To prove i), let $y' \in \underline{Y}_{k+1}$ and $u'=\pi(y')$. By the new definition of $\underline{Y}_{k+1}$, we know that there exists $x' \sim y'$ such that Eq. \eqref{eq:A2} holds with $X' = \underline{X}_{k}^\pi$. By Lemma \ref{lem:A1}, there exists $w' \in W$ such that $f(x',u',w') \in X' =  \underline{X}_{k}^\pi \subseteq X_k^\pi$, which implies that $y' \in Y_{k+1}^\pi$. To prove ii), let $x' \in \underline{X}_{k+1}^\pi$ be arbitrary and let $u' = \pi(\underline{y}_{k+1})$. By the new definition of $\underline{X}_{k+1}^\pi$, we know that $x' \sim \underline{y}_{k+1}$ and  \eqref{eq:A2} holds with $X' = \underline{X}_{k}^\pi$. 
Again, by Lemma \ref{lem:A1}, there exists $w'$ such that $f(x',u',w') \in X' =  \underline{X}_{k}^\pi \subseteq X_k^\pi$. That is, $x' \in \widetilde{X}_{k+1}^\pi$. 
\end{proof}

For switched affine systems in Eq. \eqref{eq:syst_swa}, \eqref{eq:sysm_swa} whose $B_s$ and $K_s$ matrices are linear in $p \in P = \text{cvxh}(\mathcal{V}_P)$ and $q \in Q = \text{cvxh}(\mathcal{V}_Q)$, the updated iteration \eqref{eq:iter32}, \eqref{eq:iter32}  still lead to polytopic computations similar to Eq. \eqref{eq:YkP}, \eqref{eq:XkP}. 
Then, in the forward expansion phase (Eq. \eqref{eq:a_fwd_e3}), $p_t$ is picked so that Eq. \eqref{eq:relation} holds, and then $w_t$ is picked so that $\widehat{f}_{p_t, q}\big(x_t, \pi(y_t), w_t\big) \in \underline{X}^\pi_{N-t-1}$ for all $q \in Q$, which is feasible by \eqref{eq:iter32}, \eqref{eq:iter34}. Finally,  the system evolves to the new state $x_{t+1} = f(x_t, u_t, w_t)$.

\subsection{Falsifying Vision-Based Controllers}\label{sec:ext_vision}
\begin{figure}[h]
\centering
    \includegraphics[width=0.7\textwidth]{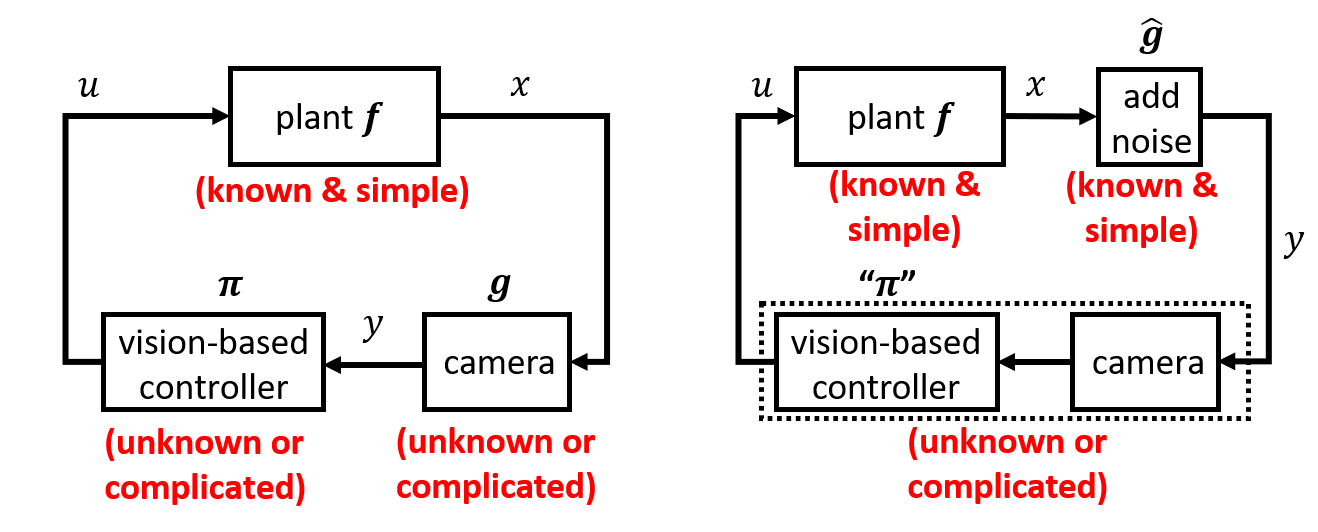}
      \caption{Block diagram of vision-in-the-loop systems with noisy images (left), and with perturbation in a low dimensional semantic space (right).}
    \label{fig:blk_1}
\end{figure}
The block diagram of a typical vision-in-the-loop system is shown in Figure \ref{fig:blk_1} (left).  
The main challenge of falsifying such systems is twofold. 
First, since $y$ represents an image, the dimension of the output space is too high \cite{balakrishnan2019specifying}, \cite{dreossi2019compositional}. 
Second, the camera is hard to model formally and hence it is impossible to perform the proposed alternating iteration (particularly Eq.  (26), (28)).  
To tackle these two challenges, we proceed as in  \cite{dreossi2019compositional} and restrict the search to a low dimensional semantic space. To this end, we consider the system shown in Figure \ref{fig:blk_1} (right) instead, where the complicated vision-based controller and the camera model are viewed as one block (i.e., the $\pi$-block), while the sensor noise is captured by a simple $\widehat{g}$-block that perturbs the real state $x$ (e.g., $y \in \widehat{g}(x) = x + \widehat{v}$, where $\widehat{v} \in \widehat{V}$ is the admissible perturbation).  
Here, the output $y$ has the same dimension as $x$ and the map $g$ is simple, and is hence amiable to our approach. 
The scheme can be interpreted as ``attacking" the system in Figure \ref{fig:blk_1} (left) by replacing the true image by an imitating one  that is taken at a perturbed state, and the perturbation can found by falsifying the system in Figure \ref{fig:blk_1}  (right).

The following proposition validates our approach by showing that, under certain Lipschiz assumption on the camera model\footnote{The real cameras are rarely Lipschiz continuous in the state, but such assumptions are used  in analysis of perception-based controllers, e.g., see \cite{dean2020guaranteeing}. }, 
the falsifying images obtained by semantic space perturbation (i.e., falsify the system with $\widehat{g}$) can be generated by the real camera model $g$. This means we can find an adversarial scenario for the real vision-in-the-loop system by falsifying $\widehat{g}$ in the loop instead. 

\begin{prop}\label{prop:vision}
\normalfont 
Let the state space $X$ be a normed space and $g(x,v) = c(x)  + v$, where $c: X \rightarrow \mathbb{R}^{n_{\rm img}}$ models a camera that maps a state to an image and $V \subseteq \mathbb{R}^{n_{\rm img}}$ is a set of additive noise in the image space. 
Assume that $c$ is globally Lipschiz (i.e., $\exists L \geq 0:  \forall x, x' \in X: \Vert c(x) - c(x') \Vert \leq L \Vert x - x'\Vert$) and that the origin is in the interior of $V$. 
Define $\widehat{g}(x, \widehat{v}): = x  + \widehat{v}$. 
There exists $\widehat{V} \subseteq X$ such that for all $x$, any image generated by $\widehat{g}$ can be also generated by $g$ under some noise $v$, i.e., $ \big\{c\big(\widehat{g}(x, \widehat{v})\big)\mid \widehat{v} \in \widehat{V}\big\}\subseteq \{g(x,v)\mid v \in V\}$. 
\end{prop}

\begin{proof}
Since $V$ contains the origin, there exists $\epsilon$ such that $\{v \in \mathbb{R}^{n_{\rm img}} \mid \Vert v \Vert \leq \epsilon\} \subseteq V$. Define $\widehat{V} = \{\widehat{v} \in X \mid \Vert \widehat{v} \Vert\leq \tfrac{\epsilon}{L}\}$, and let $x \in X$ be arbitrary. We have 
\begin{align}
& \big\{c\big(\widehat{g}(x, \widehat{v})\big) \mid \widehat{v} \in \widehat{V}\big\} = \{c(x + \widehat{v})  \mid \Vert \widehat{v} \Vert\leq \tfrac{\epsilon}{L}\} = \{c(x') \mid \Vert x' - x \Vert \leq  \tfrac{\epsilon}{L}\}  \nonumber \\
 \subseteq & \{c(x') \mid \Vert c(x') - c(x) \Vert \leq \epsilon \} = \{c(x) + v \mid \Vert v \Vert \leq \epsilon \} \subseteq \{g(x,v) \mid v \in V\}. 
\end{align}
\end{proof}
An example with an end-to-end controller can be found in Section \ref{sec:eval} (Example \ref{eg:vision}).

\subsection{Falsification against Finite-Time Reach-Avoid Specification}\label{sec:ext_reach}
Our approach can be extended to falsify systems against finite-time reach-avoid specifications defined by the following signal temporal logic (STL) formula \vspace{-0mm}
\begin{align}
\Box (x \in X_{\rm safe}) \wedge \Diamond_{\li 0, t_{\rm max}\ri} (x \in X_{\rm target}), \label{eq:reach_avoid}
\end{align}
where $\Box$ and $\Diamond$ are temporal operators ``always'' and ``eventually'', and $X_{\rm target}$ is a set of target states that we wish to reach at some time  $t \in \li 0, t_{\rm max}\ri$. For the detailed semantics of STL, we refer the readers to \cite{donze2010robust}. Eq. \eqref{eq:reach_avoid} specifies a safety property, i.e., any violation of the property occurs in finite time. 
To falsify a system against specification \eqref{eq:reach_avoid}, 
we construct the following system with an augmenting state $z$ that captures the timing dynamics: \vspace{-1.5mm}
\begin{align} 
x_{t+1} = f(x_t, u_t, w_t), \ \ z_{t+1} = z_t + 1.
\end{align}
The falsification problem can be converted to the one that has already been solved in Section \ref{sec:sol} in the $(x,z)$-space, with the initial set $\Xi_{\rm init} : = X_{\rm init} \times \{0\}$  and unsafe set $\Xi_{\rm unsafe} = \Xi_{\rm unsafe}^1 \cup \Xi_{\rm unsafe}^2$, where $\Xi_{\rm unsafe}^1 : =  X_{\rm unsafe} \times \li 0, \infty \ri$ corresponds to violation of state invariance, and $\Xi_{\rm unsafe}^2 := (X\setminus X_{\rm target}) \times \li t_{\rm max}, \infty\ri$ corresponds to missing the deadline of reaching the target set. 
To apply the proposed alternating backward expansion algorithm, we start from $\Xi_{\rm unsafe}^1$ and $\Xi_{\rm unsafe}^2$ separately. 
In particular, when expanding from $\Xi_{\rm unsafe}^2$ backwards, we also restrict $x \in X\setminus X_{\rm target}$ so that, in the falsifying trajectory, the target set is never visited before $z = t_{\rm max}$.
An example can be found in Section \ref{sec:eval} (Example \ref{eg:reach_avoid}).

\section{Examples and Evaluations}\label{sec:eval}

\liren{We present multiple examples to illustrate the efficacy of our approach and its extensions. 
The computation time for these examples are reported. 
We also evaluate the approach against randomly generated problem instances. 
}

\begin{eg}\label{eg:cont}
\normalfont
Consider a 2D discrete-time linear system in the form of Eq. \eqref{eq:syst_swa}, \eqref{eq:sysm_swa} (no switching input $s$), with $A = E = C = F =  I_2$, $B = [1 \ 0]^\top$, $D = [0 \ 0]^\top$ and $K = [0 \ -1]^\top$.  
The control set $U=[-1.2, 1.2]$, the noise set $V=[-1,1]\times[-0.1,0.1]$ and the disturbance set $W=\{w \in \mathbb{R}^2 \mid \Vert w \Vert_1 \leq 0.2\}$. 
The state moves downwards and the goal is to avoid the unsafe set $X_{\rm unsafe} = [3,17]\times [0,3]$ from the initial set $X_{\rm init}=[0,20]\times [17,20]$. 
We falsify the following four controllers whose expressions are unknown to the algorithm: 
\begin{align}
    \pi_1(y) & = \begin{cases}
    -1.2 & \text{ if } y \in [0,10]\times[0,20]\\
    1.2 & \text{ if } y \in (10,20]\times[0,20]
    \end{cases}, % \\
    \ \ \ \ \ \ \pi_2(y) = \begin{cases}
    -1.2 & \text{ if } y^2 > y^1 \\
    1.2 & \text{ otherwise } 
    \end{cases}, \nonumber \\
     \pi_3(y) & = \begin{cases}
	-1.2 & \text{ if } (y^2\leq 15-y^1 \wedge y^2\geq y^1) \vee y^1 \leq 10\\
	1.2 & \text{ otherwise }
	\end{cases}, \ \ \begin{array}{c}\pi_{\rm NN}  :  \\ \text{feedforward }\\ \text{neural net}\end{array}. 
\end{align}
where $y^1$, $y^2$ are the horizontal and vertical coordinates of $y$ respectively. 
% In this example,  the sets $\underline{Y}_k$ and $\underline{X}_k^\pi$ can be computed as polytopes as suggested by Proposition \ref{prop:swaP}, and the falsification trajectories can be constructed by Eq. \eqref{eq:a_fwd_e1}-\eqref{eq:a_fwd_e3}.
The perfect information dual winning set is inner approximated by $\bigcup_{k=1}\underline{D}_k$ (purple) where $\underline{D}_k$ is the polytopic set computed by the backward expansion algorithm in Eq.  \eqref{eq:bkwd_e1}, \eqref{eq:bkwd_e2} by ignoring the union operation. 
We then start the alternating backward expansion with $\underline{X}_0^\pi = \underline{D}_{M}\subseteq D$ where $M =\max \{k\mid \underline{D}_k \neq \emptyset\}$. 
Fig. \ref{fig:eg4} shows the obtained sets (upper) and the induced falsification runs (lower). 
The sets $\underline{X}_k^\pi$ (gray solid), $\underline{Y}_k^\pi$ (blue transparent) are  polytopes. 
We  obtain different falsification  trajectories $\textbf{x}_T$ (black solid) and output sequences $\textbf{y}_T$ (blue dotted), specific to each controller by running the same algorithm blindly. The controllers are only sampled along the output sequences $\textbf{y}_T$. 
\hfill\(\dagger\)
\end{eg}

\begin{figure}[h]
 \centering
    \includegraphics[width=0.8\textwidth]{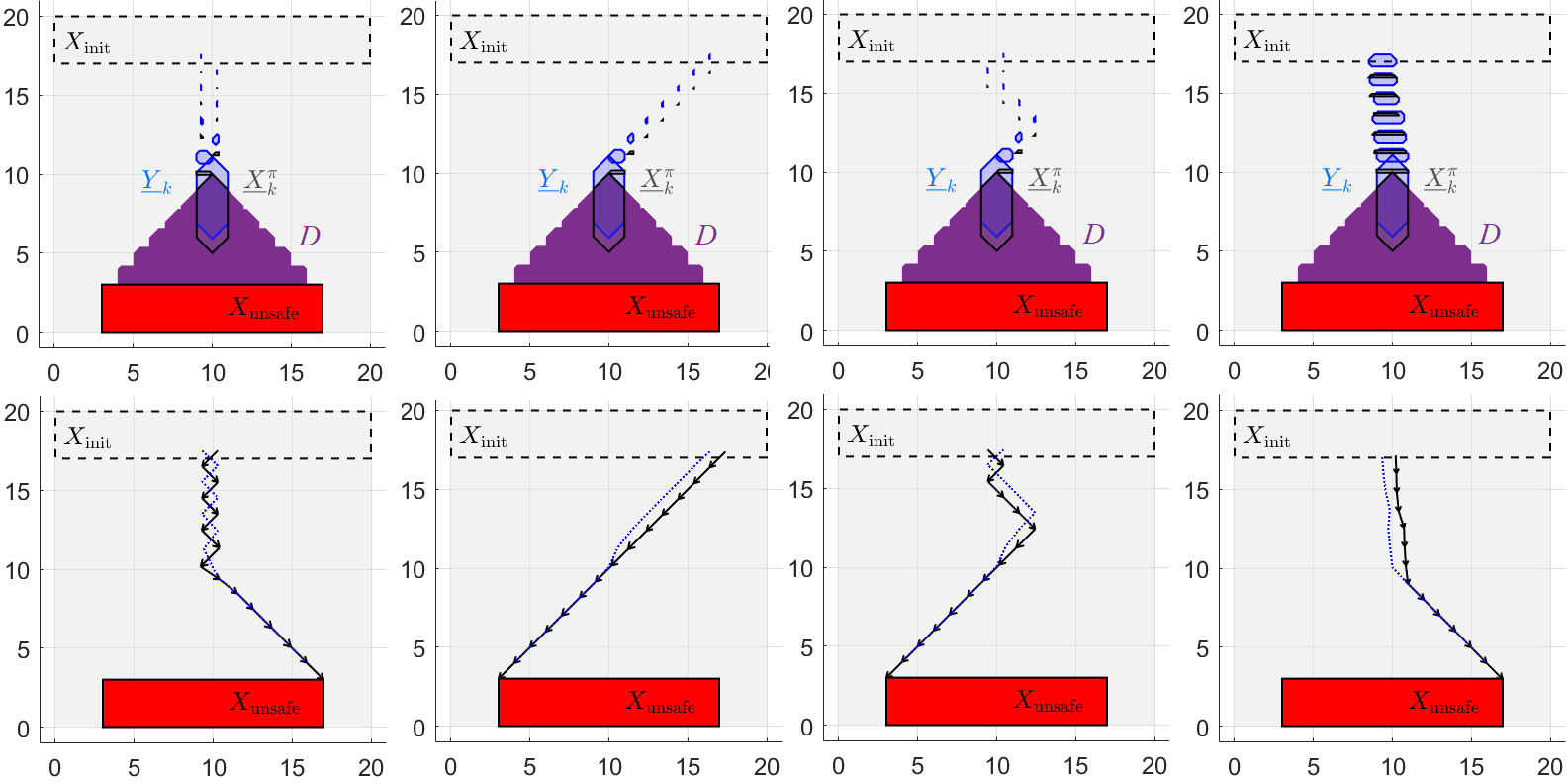}
    \caption{Example \ref{eg:cont}, falsification under controller $\pi_1$ (1$^{\rm st}$ column),   $\pi_2$  (2$^{\rm nd}$ column), $\pi_3$  (3$^{\rm rd}$ column) and $\pi_{\rm NN}$  (4$^{\rm th}$ column). }
    \label{fig:eg4}
\end{figure}

One notable and interesting phenomenon is that the controller-specific adversarial scenarios where the imperfect information plays a key role usually occur near the ``decision boundary'' of the controller, where a discontinuous decision is made. 
In Example \ref{eg:cont}, to cross the obstacle $X_{\rm unsafe}$, a controller needs to decide to go either left or right. In a typical adversarial scenario, the true state is on one side but is also very close to the decision boundary, while the sensor noise makes the controller mistakenly believe that the state is on the other side and the ``opposite'' control action should be taken. This results in chattering behavior along the  decision boundary until the state hits the dual winning set $D$, from where safety violation is assured by the adversarial process noise only. 
To summarize, such  scenarios can be found at places where the closed-loop system's vector field is discontinuous and these scenarios are hence controller-specific. 
% \todo{maybe we should more carefully state this. what is not discontinuous if i  only  control switching?}
Our approach provides a way to detect such critical cases without knowing the explicit expression of the controllers. 
% \vspace{-0.5mm}

Note that such decision boundary is in general inevitable given static controllers and an operational domain that is not simply connected. 
For example, it is pointed out in \cite{sontag1999stability} that the region of attraction of an asymptotically stable equilibrium cannot have holes if the closed-loop system's vector field is continuous. Therefore, the operational domain of a stabilizing controller having holes (i.e., $X_{\rm unsafe}$) will imply discontinuous closed-loop dynamics, and hence the existence of a ``decision boundary".

\begin{eg}\label{eg:lqr_unknown}
\normalfont 
(Section 4.3)
We provide an example where the control space refinement is necessary to falsify the system. 
The system is in the form of Eq. \eqref{eq:syst_swa}, \eqref{eq:sysm_swa} (no switching input $s$), with 
\begin{align}
& A = 
\left[
\begin{array}{cc}
0.9745 & 0.2132 \\
0.002547  & 1.151 \\
\end{array}
\right],  
\ \  B = \left[\begin{array}{c} 0.01959 \\0.1961  \end{array}\right], 
\ \  E = \left[\begin{array}{c} 0.01959 \\-0.04509  \end{array}\right],  
\end{align}
$C = F =  I_2$, $D = K = [0 \ 0]^\top$. 
The control set $U=[-1.5,1.5]$, the noise set $V=[-0.1,0.1]\times[-0.1,0.1]$ and the disturbance set $W=[-0.2,0.2]$. The unsafe set $X_{\rm unsafe} = [1,2]\times[-0.5,0.5]$ and the initial set $X_{\rm init}=[-0.15,0.15]\times[-0.15,0.15]$. 
The system is stabilized by an LQR controller with saturation $\pi_{\rm LQR} = \min\{\max\{[-1\ -1]x, -1.5\}, 1.5\}$ (unknown to the falsification algorithm). 
Fig. \ref{fig:eg_lqr_unknown} shows the falsifying trajectory (left) found by Algorithm \ref{alg:main}. 
This falsifying scenario is nontrivial because the uncertainty profiles (right) need to follow a special periodic-pattern. 
If the control space refinement were not used, the alternating iteration would terminate prematurely before reaching the initial set backwardly. Essentially, $\pi_{\rm LQR}$ is falsified because it is not using the full capability allowed by the control set $U = [-1.5, 1.5]$.  
\hfill\(\dagger\)
\end{eg}
\begin{figure}[h]
 \centering
\vspace{-0mm}
    \includegraphics[width=0.9\textwidth]{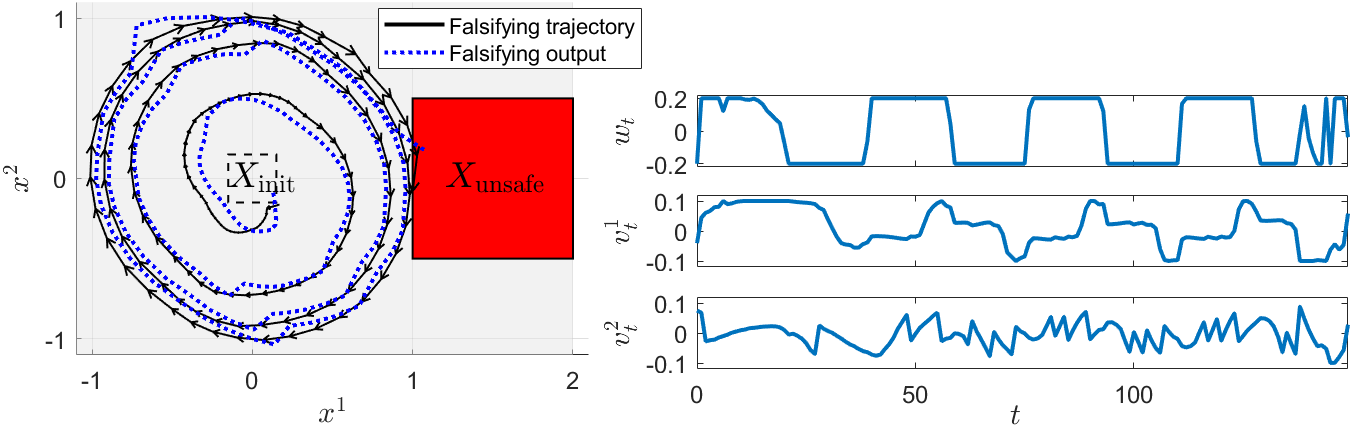}
\vspace{-0mm}	
    \caption{Example \ref{eg:lqr_unknown}: falsifying trajectory (left), disturbance and noise profiles (right). }
    \label{fig:eg_lqr_unknown}
\vspace{-0mm}	
\end{figure}

\begin{eg}\label{eg:swa}
\normalfont 
We apply our approach to the following buck converter system (adopted from \cite{yang2019correct}). 
\begin{figure}[h]
    \includegraphics[width=0.55\textwidth]{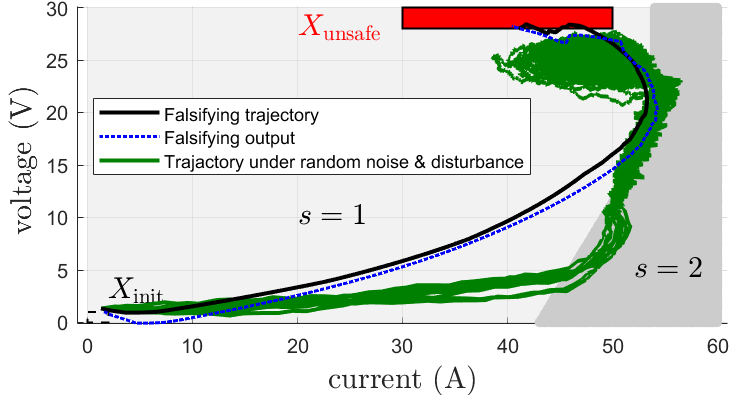}	
      \caption{Example \ref{eg:swa}. }
    \label{fig:eg_swa}
\end{figure}
The system's states are the current and the voltage of the converter. 
The system has two discrete inputs only, i.e., $s \in \{1,2\}$, corresponding to the converter's dynamics regularized with two different on-off sequences. Particularly, 
\begin{align}
& A_1 = 
\left[
\begin{array}{cc}
0.9999 & -0.04980 \\
0.003984 & 0.9919 \\
\end{array}
\right],
K_1= \left[\begin{array}{c} 1.250\\0.002618 \end{array}\right],  
 A_2 = 
\left[
\begin{array}{cc}
0.9992 & -0.1360 \\
0.01088 & 0.9775 \\
\end{array}
\right], 
K_2 =  \left[\begin{array}{c} 1.250\\0.006948\end{array}\right],  
\end{align}
and $E_1 = E_2 = C = F = I_2$, $B = [0 \ 0]^\top$. 
The noise set $V = [-1,1]\times[-1,1]$ and the disturbance set $W = [-0.2, 0.2]\times[-0.2, 0.2]$. 
The initial set  $X_{\rm init} = [0,2]\times[0,1]$. 
The unsafe set $X_{\rm unsafe} = [30,50]\times[28,30]$ corresponds to the overshoot in the converter's voltage. 
The control objective is to reach a target voltage level ($25$ V) and stay close indefinitely. To this end, the following rule-based switching controller is used: 
\begin{align}
\pi(y) = \begin{cases}
2 & \text{ if } \big((\widehat{y}^2 \leq23.5) \wedge(-23\widehat{y}^1+14\widehat{y}^2 \leq -989)\big)  \vee \big((\widehat{y}^2 > 23.5)\wedge (\widehat{y}^1 \geq 54)\big) \\
1 & \text{ otherwise } \\
\end{cases}
\end{align}
where $\widehat{y} =\min\big\{\max\{y,[0 \ 0]^\top\},[60 \ 25]^\top\big\}$.   
The switching rule is depicted in Fig. \ref{fig:eg_swa}: $s = 1$ is used in the light gray area whereas $s = 2$ is used in the dark gray area. 
There is a decision boundary along which the closed-loop dynamics is discontinuous. However, different from Example \ref{eg:cont}, the vector field points towards this decision boundary. 
Fig. \ref{fig:eg_swa} shows ten closed-loop trajectories (green solid) under random disturbance and noise. 
These  trajectories tend to first reach the decision boundary, chatter and slide along it until the target voltage level is achieved,
and eventually leave the decision boundary and stay close to the target state. Particularly, the unsafe region is avoided in all the ten simulations. 
Due to the chattering behavior along the switching surface, it is in general difficult to use verification tools (e.g., {\tt{flow*}} \cite{chen2013flow}) based on reachable set over-approximation. We are able to find a falsifying trajectory (black solid) using our approach, which has only query access to $\pi$. Unlike Example \ref{eg:cont}, this falsifying trajectory tends to avoid the switching surface. \hfill\(\dagger\)
\end{eg}

\begin{eg}\label{eg:car}
\normalfont  (Section \ref{sec:ext_mismatch}) 
We illustrate our approach's capability to deal with model mismatch by an intersection management problem with two autonomous cars. 
% The goal is to show its capability of handling % practical problems with 
% systems with relatively high dimensional \textit{output} space. 
The considered scenario is illustrated with Fig. \ref{fig:eg_car}. Two identical autonomous cars approach an uncontrolled intersection at the same time, and coordinate with each other to cross the intersection safely (i.e., without collision) and optimally (i.e., minimizing the time of crossing) by solving the same hybrid MPC problem on their own computational platform. A rough model of the other car is known to each car so that they can predict the other's behavior and react accordingly. 
The states $ [h_1, d_1, h_2, d_2]^\top$, where $h_i$ is the $i^{\rm th}$ car's position relative to the center of the intersection  and $d_i$ is its speed, evolve with the following dynamics \cite{nilsson2015correct}: 
\begin{align}
\dot{h}_i = d_i + \delta d_i, \ \ \dot{d}_i = \tfrac{1}{m}(F_i - f_{0} - f_{1}d_i - f_{2}d_i^2 + \delta F_i).  
\end{align}
The control inputs $[F_1, F_2]^\top$ are the forces acting on the two cars. 
The disturbances $[ \delta d_1,  \delta F_1, \delta d_2,  \delta F_2]^\top$ are state dependent because they cannot move the car backwards. 
The above continuous-time dynamics is discretized with $0.2$s and linearized. 
The nonlinear term $f_{2}d_i^2$ is captured with an additive model uncertainty $p$, whereas $q$ is not needed in this example. 
The detailed parameters can be found in Table. \ref{tab:car_para}. 
The unsafe set $X_{\rm unsafe} = [-2,2]\times[5,20]\times[-2,2]\times[5,20]$ corresponds to both cars being within $2$m to the intersection center but the speeds being both larger than $5$m/s. 
The initial set $X_{\rm init} = [-50, -40]\times[0,20]\times[-50, -40]\times[0,20]$. 
The observation $y$ consists of an estimate of the two cars' states. Here, the two cars do not communicate, so they have a better estimation of their own states but a coarser estimate of the other car's state. Hence $y \in \mathbb{R}^8$ and $y = [y^1, y^2]$ where $y^i=[\widehat{h}^i_1, \widehat{d}^i_1, \widehat{h}^i_2, \widehat{d}^i_1]^\top = x  + v^i$ is the observation of car $i$. 
We assume that  $v^1 \in V^1 = [-0.25,0.25]\times[-0.25, 0.25]\times[-1,1]\times[-0.5, 0.5]$ and $v^2 \in V^2 = [-1,1]\times[-0.5,0.5]\times[-0.25,0.25]\times[-0.25, 0.25]$. 
The MPC controller $\pi_{\rm MPC}$ to falsify essentially contains two identical MPC formulations, one is solved on car 1 with observation $y^1$, and the other solved on car 2 with observation $y^2$. To minimize the crossing time,  
$\pi_{\rm MPC}$ will assign priority to the car that may cross the intersection sooner and ask the other car to slow down to avoid collision. Under randomly generated process and measurement noise,  
%  within their bounds, 
the controller experiences no safety violation in dozens of simulations. 
% \todo{how many simulations, what  does "always" mean}. 
Our approach finds a falsification trajectory where a collision occurs. Fig. \ref{fig:eg_car} plots the true and observed position of the two cars.  The falsification scenario is interpreted as follows: the positions and the velocities of two cars are identical, but the observation makes each car believe that it is the one with priority. Thus both cars accelerate to the intersection and collide. 
\begin{table}[h]
  \caption{Parameter values of the model in Examples 5}
\centering
{\small	\begin{tabular}{c | c | c }
    	% \toprule
	\hline
	Parameter & Description & Value \\
	\hline\hline
	 $m$  & mass of cars &  1462 (kg) \\
\hline
       $f_0$  & friction/drag term &   51  (N)  \\
\hline
	$f_1$  &   friction/drag term  &  1.2567  (Ns/m) \\
\hline
	$f_2$  &   friction/drag term  &  0.4342(N$\text{r}^2$/$\text{m}^2$)   \\
\hline
	\end{tabular} 
}
\label{tab:car_para}
\end{table}
\hfill\(\dagger\)
\end{eg}

\begin{figure}
\centering
    \includegraphics[width=0.5\textwidth]{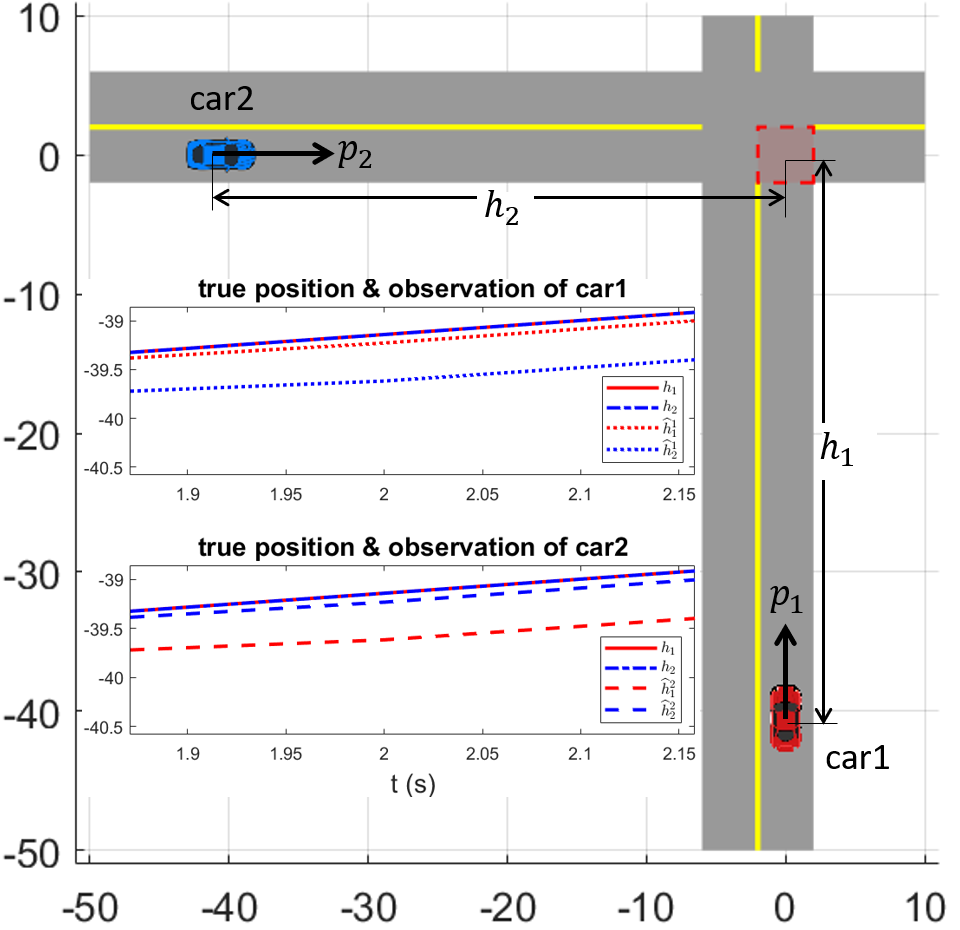}
      \caption{Intersection management: two autonomous cars.}
\label{fig:eg_car}
\end{figure}

\begin{eg}\label{eg:vision}
\normalfont  (Section \ref{sec:ext_vision}) 
Consider a self-driving car with linear lateral and longitudinal dynamics in the form of Eq. \eqref{eq:syst_swa}, \eqref{eq:sysm_swa} (no switching input $s$), with 
\begin{align}
A & = \left[
\begin{array}{cc}
1 & 0\\
0 & 1\\
\end{array}
\right], 
B = \left[
 \begin{array}{c}
0.1\\
0\\
 \end{array}
\right],
% \\
E = \left[
\begin{array}{ccc}
0.1 & 0\\
0 & 0.1\\
\end{array}
\right], 
K = \left[
\begin{array}{c}
0\\
1.5\\
\end{array}
\right]. 
\end{align}
The control set $U = [-1.8,1.8]$ and the disturbance set $W = [-0.1, 0.1]\times[-0.1, 0.1]\times[-0.05, 0.05]$.  
The observation is a camera image that is generated by {\tt{Carla}} \cite{Dosovitskiy17}. 
% The system is similar to that in Example \ref{eg:cont}, except that the dynamics along the horizontal direction is a double integrator, and the control input is the horizontal acceleration.  
An end-to-end controller $\pi_{\rm NN}$ is implemented by a feedforward neural network, whose input is a low resolution image and output is the control $u$. The approximated noise set $\widehat{V} = [-0.17,0.17]\times[-0.05,0.05]$. 
The controller's objective is to avoid a steady red truck ($X_{\rm unsafe} = [-1.5, 1.5]\times[0,5]$) from the initial set $X_{\rm init} = [-2,2]\times[-30,-29]$ without deceleration. 
The neural network is trained with a set of demonstrations (i.e., image-control input pairs) generated by a hybrid MPC controller \cite{levine2016end}. 
The obtained falsifying scenario is shown in Figure \ref{fig:eg_vision} (left). 
The green solid line represents a typical closed-loop trajectory under random noise and disturbance. In this case the truck (red box) is avoided. 
However, a falsifying trajectory can be found in the middle of the domain. 
Since the neural  network tries to interpolate the behavior of the demonstrating MPC 
\begin{figure}
\centering
\vspace{-0mm}
    \includegraphics[width=0.55\textwidth]{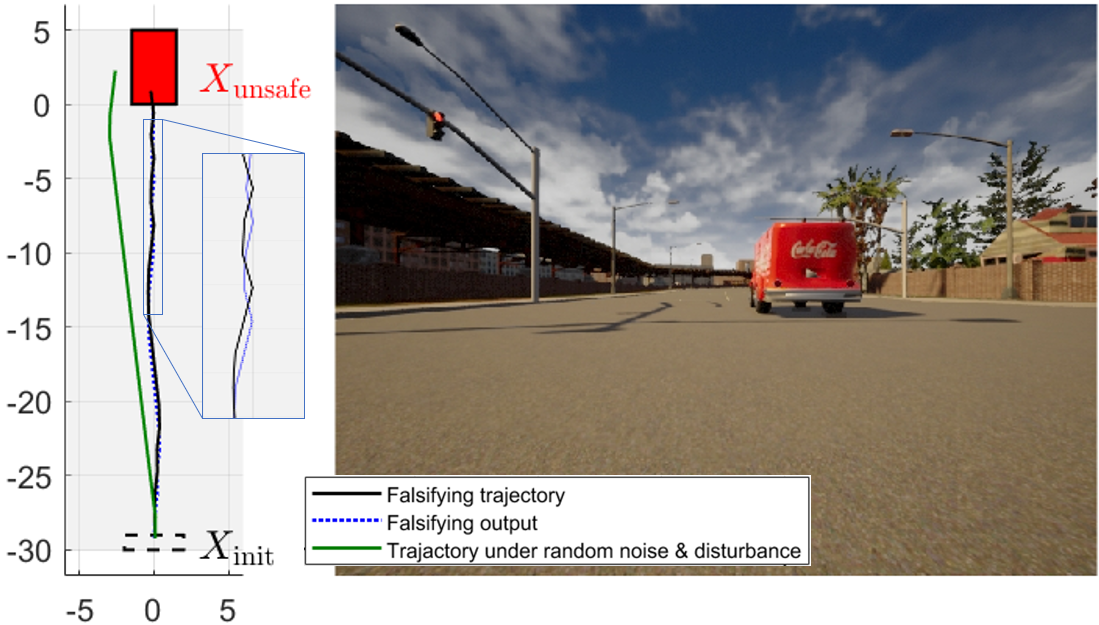}
	\vspace{-0mm}
      \caption{Example \ref{eg:vision}. Left: falsifying trajectory. Right: a typical image generated by Carla along the trajectory.}
    \label{fig:eg_vision}
 \vspace{-0mm}
\end{figure}controller, it tends to choose a positive lateral speed if it is more likely to cross the obstacle from the right, and choose a negative lateral speed otherwise. Therefore, although $\pi_{\rm NN}$ is continuous, it has a sharp change of decision in the middle of the domain, and small noise can cause undesired behavior (see the zoomed-in plot).  
A video can be found at \url{https://youtu.be/QeBoqsdmCU4}. 
\hfill\(\dagger\)
\end{eg}

\begin{eg}\label{eg:reach_avoid}
 \normalfont  (Section \ref{sec:ext_reach}) 
Consider a 2D linear system in the form of Eq. \eqref{eq:syst_swa}, \eqref{eq:sysm_swa}. 
\begin{figure}[b]
% \centering
    \includegraphics[width=0.55\textwidth]{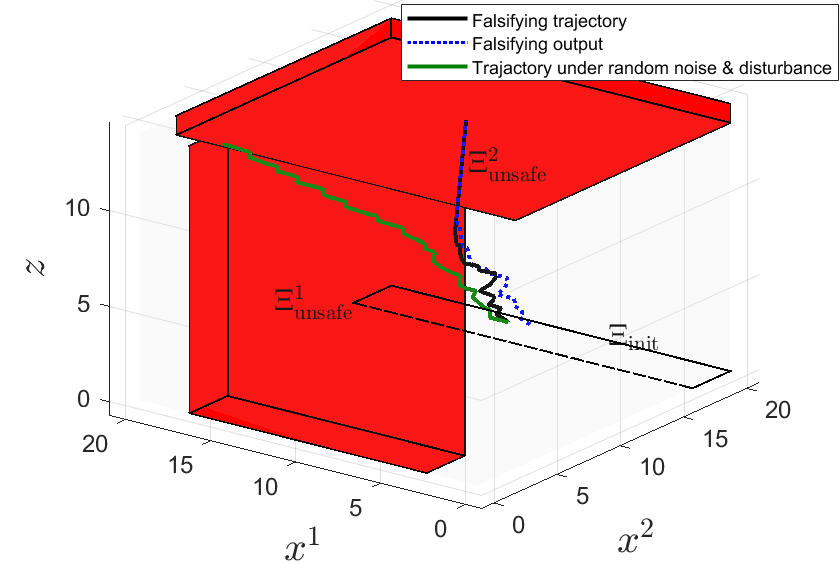}
      \caption{Example \ref{eg:reach_avoid}.}
    \label{fig:eg_reach_avoid}
 \end{figure}
%\end{wrapfigure}
where $A = I$, $B = E = 0.25I$ and $K = [0 \ 0]^\top$. The control set $U = [-1,-1]\times[-2,2]$, the disturbance set $W = [-0.1,0.1]\times[-0.3, 0.3]$, and the noise set $V = [-1,1]\times[-1,1]$. The initial set $X_{\rm init} = [0,20]\times[17,20]$, the unsafe set $X_{\rm unsafe} = [3,17]\times[0,3]$ and the target set $X_{\rm target} = [0,20]\times[0,3]$. 
The settings are similar to Example \ref{eg:cont} except that now there is i) a target set and ii) a ``braking'' action, which allows the state to stay away from the unsafe set.
However, braking will prevent the state moving towards the target set,  and the deadline $t_{\rm max}$ for reaching the target may be missed. We use a hybrid MPC controller $\pi_{\rm MPC}$ that minimizes the sum of distance from the predicted states to the target set while enforcing $x_t \notin X_{\rm unsafe}$ and $x_{t^\star}\in X_{\rm target}$ for some $t^\star \in \li 0, t_{\rm max}\ri$. 
 The controller is static and decides to go left/right and accelerate/decelerate at each point. 
A falsification scenario is found via backward expansion from set $\Xi_{\rm unsafe}^2$. 
Figure \ref{fig:eg_reach_avoid} shows the results. 
The green solid line represents a typical trajectory under random noise and disturbance, which satisfies the finite reach-avoid specification. The first half of the falsifying trajectory (black solid) ``chatters'' in the middle of the state space under the adversarial observation sequence (blue dotted), until the state is too close to  $X_{\rm unsafe}$ to cross it. The controller then decides to brake and 
cannot reach the target set by the deadline $t_{\rm max}$. This corresponds to the second half of the falsifying trajectory hitting $\Xi_{\rm unsafe}^2$. 
\hfill\(\dagger\)
\end{eg}

\subsection{Computation Time for the Examples \ref{eg:cont}-\ref{eg:reach_avoid}}
In this short section, we report the computation  time to find a satisfying adversarial scenario for Examples \ref{eg:cont}-\ref{eg:reach_avoid} in the paper. 
All the experiments are run on a 1.80 GHz machine with 16 GB RAM. Table \ref{tab:eva1} summarizes the result for each example. The second column (Dim $x$) is the dimension of the state space $X$, the third column (Dim $y$) is the dimension of the observation space $Y$, and the fourth column $T$ is the length of the obtained safety-violating trajectory. 
The computation time grows with the dimension of the state space and observation space, and also the length of the falsifying trajectory.  
We also apply the falsification tool {\tt{s-taliro}} to some examples. 
Except for Example \ref{eg:cont} with controller $\pi_2$,  {\tt{s-taliro}}
could not find any falsifying trajectories after running for hours (e.g., $10^5$ iterations for Example \ref{eg:lqr_unknown}). 
The results for {\tt{s-taliro}} is given in the last column, where $l$ is the number of iterations, $t$ is the time and $r$ is the robustness value ($r < 0$ means that the system is falsified). 
For Example \ref{eg:car} and \ref{eg:reach_avoid}, simulations with the hybrid MPC controller in the loop are very time-consuming,  hence we could only run small number of iterations.

We note that a comparison with {\tt{s-taliro}} 
is not fair in that {\tt{s-taliro}} works with black-box dynamics, hence aims to solve a harder problem. Our purpose with this comparison is to show that the obtained adversarial examples are nontrivial to find and to demonstrate the importance of using the plant model whenever it is available. 
\begin{table}[t]
  \caption{Evaluation of Algorithm \ref{alg:main} against Examples \ref{eg:cont}-\ref{eg:reach_avoid}}
\label{tab:eva1}
\centering
{\small	\begin{tabular}{c || c | c | c | c | r   l | c }
	\hline
	Example  & Dim $x$     &   Dim $y$  & $T$ &  Controller type &  CPU & time (s)  & {\tt{s-taliro}} \\
\cline{6-7}
	 index  &       &    & (\# steps) &    &  \multicolumn{1}{|c|}{Backward} &  \multicolumn{1}{|c|}{Forward} & ($l$, $t$, $r$) \\
\hline
\hline
      &   &  & 15 & Rule-based $\pi_1$ & \multicolumn{1}{|c|}{18.56} &  \multicolumn{1}{|c|}{5.09}   & ($10^4$,4226, 5.1)\\
      \cline{4-8} 
      \ref{eg:cont}   & $2$  & $2$ & 15 & Rule-based $\pi_2$ & \multicolumn{1}{|c|}{17.95} &  \multicolumn{1}{|c|}{5.06}  & ($10^4$,1098,-3e-3)\\
	\cline{4-8} 
       &    &  & 15 & Rule-based $\pi_3$  & \multicolumn{1}{|c|}{17.41} &  \multicolumn{1}{|c|}{4.56}  & ($10^4$,3085,5.2)\\
	\cline{4-8} 
      &    & & 14 & Neural net $\pi_{\rm NN}$ & \multicolumn{1}{|c|}{17.08} &  \multicolumn{1}{|c|}{4.30}  & ($10^3$,5768,0.11)\\
	\hline
	\ref{eg:lqr_unknown}   & $2$  & $2$ & 115  & LQR $\pi_{\rm LQR}$ &   \multicolumn{1}{|c|}{4096.28} &  \multicolumn{1}{|c|}{38.85} & ($10^5$,1.9e4,0.8)\\ % 19511
	\hline 
	\ref{eg:swa}   & $2$  & $2$ &132  & Rule-based & \multicolumn{1}{|c|}{263.15} &  \multicolumn{1}{|c|}{46.63} & ($10^4$,1e4,4.1) \\ % 10491
	\hline 
	\ref{eg:car}   & $4$  & $8$ & 27 & Hybrid MPC $\pi_{\rm MPC}$ & \multicolumn{1}{|c|}{1640.45} &  \multicolumn{1}{|c|}{83.12} & (50, 4.6e4, 29.4)\\
	\hline 
	\ref{eg:vision}   & $2$  & image: $3072$ & 28 & Neural net $\pi_{\rm NN}$ & \multicolumn{1}{|c|}{39.41} &  \multicolumn{1}{|c|}{1.44} & - \\
	    &   & semantic: $2$ &  & (end-to-end) & \multicolumn{1}{|c|}{ } &  \multicolumn{1}{|c|}{ } & \\
	\hline 
      \ref{eg:reach_avoid}   & $3$  & $2$ & 50 & Hybrid MPC $\pi_{\rm MPC}$ & \multicolumn{1}{|c|}{182.41} &  \multicolumn{1}{|c|}{67.25} & ($10^2$,4.4e4,0.75)\\  
	\hline 
	\end{tabular} 
}
\end{table}

\subsection{Evaluation with Randomly Generated Instances}

We generate random problem instances and test the efficacy and scalability of our approach. 
We consider integrator-like dynamics of different dimension with small model mismatch. The output map is assumed to be $y = x + v$. 
The tested 2D systems are single integrators on a plane, with small randomly generated model mismatch. 
The mismatch is in the form of an additive quadratic term in the state $x$, which is bounded by a small number.  
The 6D systems are double integrators in a 3D space, and the 10D systems are double integrators in a 3D space with a 4D stable but uncontrollable subspace. 
In all problem instances, the objective is to reach a polytopic target set from a randomly generated rectangular initial set $X_{\rm init}$, while avoiding multiple random rectangular sets ($X_{\rm unsafe}$ is the union of these rectangles). 
Since the safe domain is nonconvex, we use hybrid MPC controllers.

Table \ref{tab:eva2} shows the results. 
In the obtained adversarial scenarios, since the trajectories hit $X_{\rm unsafe}$ by construction, the falsification is considered successful  (marked as ``\checkmark'') if the safety-violating trajectory starts from $X_{\rm init}$, 
whereas the ``$\times$'' mark indicates that the obtained safety-violating trajectory is from $x_0 \notin X_{\rm init}$. 
For each problem instance, we also run ten simulations from the same initial state $x_0$ under random noise and disturbance profiles.  The number of violations among these ten random simulations is reported in the fifth column of Table \ref{tab:eva2}.    
While the initial state $x_0$ is not always in $X_{\rm init}$, random simulations therefrom still tend to be safe in many cases (e.g., 2D case, instances 4, 6, 8-10, 10D case, instance 2).  
This suggests that the scenarios we found, though not from $X_{\rm init}$, are still nontrivial. 
Moreover, the time of finding an adversarial scenario is small comparing to the time of random simulations. 
One reason is because the hybrid MPC controller  solves a sequence of MILPs, which is time consuming.
In fact, a considerable portion of the falsification time is devoted to the query step. 
While the  time for set manipulation grows with the plant's dimension $n_x$, the overall falsification time also depends on the length of the falsifying trajectory and the time for controller query. 
Our zonotope-based implementation works for 10D dynamics, while the polytope-based implementation (using MPT3) can hardly go beyond 4D.

\begin{table}[]
  \caption{Evaluation of Algorithm \ref{alg:main} against Random Instances}
\label{tab:eva2}
\centering
{\small	
\begin{tabular}{c|c|c|c|c|c}
\hline
               Dim $x$   &  Instance & success & falsification & \# violations &  CPU time of\\  
                ($n_x$)  & index & (from $X_{\rm init}$) & CPU time& under random $v,w$  & 10 simulations \\ 
                  & &  &  (s)  & (out of 10)  &  (s)   \\ \hline\hline
2 & 1 & \checkmark & 874 & 1 & 6799 \\  \cline{2-6} 
(polytope)  & 2 & \checkmark & 1018 & 0 & 8571 \\ \cline{2-6} 
 & 3 & \checkmark & 719& 0 & 9275 \\ \cline{2-6} 
 & 4 & $\times$ & 650 & 0 & 4917 \\ \cline{2-6} 
  & 5 & \checkmark & 1278 & 0 & 10970 \\ \cline{2-6} 
  & 6 & $\times$ & 663  & 0 & 7628 \\ \cline{2-6} 
  & 7 & \checkmark & 695 & 0 & 19470 \\ \cline{2-6} 
  & 8 & $\times$ & 639 & 0 & 17164 \\ \cline{2-6} 
  & 9 & $\times$ & 766 & 0 & 16329 \\ \cline{2-6} 
  & 10 & $\times$ & 1716 & 0 & 23420 \\ \cline{2-6} 
  & 11 & \checkmark & 1991  & 3 & 14952 \\ \hline
		6 & 1 & \checkmark & 1860 & 0 & 80283 \\ \cline{2-6} 
		 (zonotope) & 2 &  \checkmark &  1602  & 0 & 21957 \\  \cline{2-6} 
		  & 3& \checkmark  & 281  & 0 & 26205 \\ \cline{2-6} 
               & 4 & $\times$ & 745 & 2 & 8281\\ \cline{2-6} 
               & 5 & \checkmark & 2841 &  0 & 17920\\ \hline
		10 & 1 & \checkmark & 542  & 0 & 22026 \\ \cline{2-6}
	(zonotope)	  & 2 & $\times$ & 1045 & 0 & 10147 \\ \hline
\end{tabular}
}
\end{table}

\section{Conclusion and Future Work}
We considered the problem of falsifying safety properties for systems with static controllers that do not have access to perfect state information. 
A synthesis-guided falsification approach was proposed to find a class of ``interesting" falsifying scenarios specific to the given controller. 
The approach was shown to be applicable to systems with relatively simple dynamics but complicated or even unknown controllers, including MPC and neural network controllers. We presented extensions of our approach to falsify vision-in-the-loop systems against finite-time reach-avoid specifications. 
Some computational aspects of the proposed approach were also discussed. 
We believe that our approach nicely complements the black-box approaches from falsification and white-box approaches from adversarial learning. 
For future work, we will explore the falsification problem i) against more complicated specifications 
and ii) for other classes of dynamics, either using abstraction-based techniques or nonlinear backward reachable set computations.

\begin{acks}
This work is supported in part by ONR grant \# N00014-18-1-2501.
\end{acks}

\bibliographystyle{ACM-Reference-Format}
\bibliography{main}

\end{document}